\documentclass[letterpaper,USenglish,cleveref]{lipics-v2021}

\pdfoutput=1
\hideLIPIcs 

\usepackage{macros}
\usepackage{lineno}

\title{Simple Circuit Extensions for XOR in PTIME}

\author{Marco Carmosino}{MIT-IBM Watson AI Lab, Cambridge, MA, USA}{mlc@ibm.com}{https://orcid.org/0009-0007-1118-1352}{}

\author{Ngu Dang}{Boston University, Boston, MA, USA}{ndang@bu.edu}{https://orcid.org/0009-0004-2774-2247}{}

\author{Tim Jackman}{Boston University, Boston, MA, USA}{tjackman@bu.edu}{https://orcid.org/0000-0002-2293-5670}{}

\authorrunning{M. Carmosino, N. Dang, and T. Jackman}

\Copyright{Marco Carmosino, Ngu Dang, and Tim Jackman}



\category{} 

\relatedversion{} 
\acknowledgements{}

\begin{document}
\nolinenumbers

\maketitle

\begin{abstract}
 The Minimum Circuit Size Problem for Partial Functions ($\MCSP^*$) is hard assuming the Exponential Time Hypothesis (ETH) (Ilango, 2020). This breakthrough hardness result leveraged a characterization of the optimal $\{\land, \lor, \neg\}$ circuits for $n$-bit $\OR$ ($\OR_n$) and a reduction from the partial $f$-Simple Extension Problem where $f = \OR_n$. It remains open to extend that reduction to show ETH-hardness of total $\MCSP$.  However, Ilango observed that the total $f$-Simple Extension Problem is easy whenever $f$ is computed by read-once formulas (like $\OR_n$).  Therefore, extending Ilango's proof to total $\MCSP$ would require one to replace $\OR_n$ with a slightly more complex but similarly well-understood Boolean function.
    
  This work shows that the $f$-Simple Extension problem remains easy when $f$ is the next natural candidate: $\XOR_n$. We first develop a fixed-parameter tractable algorithm for the $f$-Simple Extension Problem that is efficient whenever the optimal circuits for $f$ are (1) linear in size, (2) polynomially ``few'' and efficiently enumerable in the truth-table size (up to isomorphism and permutation of inputs), and (3) all have constant bounded fan-out. $\XOR_n$ satisfies all three of these conditions. When $\neg$ gates count towards circuit size, optimal $\XOR_n$ circuits are binary trees of $n-1$ subcircuits computing $(\neg)\XOR_2$ (Kombarov, 2011). We extend this characterization when $\neg$ gates do not contribute the circuit size. Thus, the $\XOR$-Simple Extension Problem is in polynomial time under both measures of circuit complexity.

  We conclude by discussing conjectures about the complexity of the $f$-Simple Extension problem for each explicit function $f$ with general circuit lower bounds over the DeMorgan basis. Examining the conditions under which our Simple Extension Solver is efficient, we argue that \emph{multiplexer} functions ($\MUX$) are the most promising candidate for ETH-hardness of a Simple Extension Problem, towards proving ETH-hardness of total $\MCSP$.
\end{abstract}

\newpage

\section{Introduction}\label{sec:intro}
Circuits model the computation of Boolean functions on fixed input lengths by acyclic wires between atomic processing units --- logical ``gates.''  To measure the circuit complexity of a function $f$, we first fix a set of gates $\cB$ --- called a \emph{basis}. This work studies circuits over the following basis: fan-in 2 AND, fan-in 2 OR, and fan-in 1 NOT gates. We consider two complexity size measures $\mu_{\cD}$ and $\mu_{\cR}$, which count \emph{only} the binary gates and the \emph{total} number of gates in a circuit respectively. We will refer to $\cB$ equipped with these two complexity measures as $\cD$, the DeMorgan basis, and $\cR$, the Red'kin basis respectively\footnote{We will specify a basis if a statement pertains to \emph{only} that basis. If the basis is not specified, then the statement applies to both $\cD$ and $\cR$.}.

Basic questions about these models have been open for decades; we cannot even rule out the possibility that every problem in $\NP$ is decided by a sequence of linear-size circuits (see page 564 of \cite{Jukna2012}). Despite this, the ongoing search for circuit complexity lower bounds has fostered rich and surprising connections between cryptography, learning theory, and algorithm design \cite{HiraharaS2017, GolovnevIIKKT2019, Santhanam2020, CheraghchiKLM2020, RenS2021, HuangIR2023}. The \emph{Minimum Circuit Size Problem} ($\MCSP$, \cite{KabanetsC00}) appears in all of these areas, asking:
\begin{quote}
    Given an $n$-input Boolean function $f$ as a $2^n$-bit truth table, what is the minimum $s$ such that a circuit of size $s$ computes $f$ ?
\end{quote}
The \textbf{existential} question --- do functions that require ``many'' gates exist? --- was solved in 1949: Shannon proved that almost all Boolean functions require circuits of near-trivial\footnote{From using a lookup table.} size $\Omega(\frac{2^n}{n})$ by a simple counting argument \cite{Shannon49}.  The current best answer to the \textbf{explicit} question in the DeMorgan basis --- is such a hard function in $\NP$? --- is a circuit lower bound of $5n - o(n)$, proved via \emph{gate elimination} \cite{IwamaM02}. This is far from the popular conjecture that $\NP$-complete problems require super-polynomial circuit size.

The \textbf{algorithmic} question --- is $\MCSP$ $\NP$-hard? --- remains open after nearly fifty years \cite{Trakhtenbrot84}, even under strong complexity assumptions such as the Exponential Time Hypothesis (ETH). But, many natural variants of $\MCSP$ have been proven $\NP$-hard unconditionally. For instance, $\DNF$-$\MCSP$ \cite{Masek1979}, $\MCSP$ for $\OR$-$\AND$-$\MOD$ Circuits \cite{HiraharaOS18}, and $\MCSP$ for multi-output functions \cite{IlangoLO20} are now known to be $\NP$-hard. Furthermore, $\MCSP$ for partial functions ($\MCSP^*$) \cite{Ilango20} is hard under the Exponential Time Hypothesis (ETH), later extended to unconditional $\NP$-hardness under randomized reductions \cite{Hirahara22}.

Our work studies the feasibility of generalizing Ilango's technique for ETH-hardness of $\MCSP^*$ to total $\MCSP$. In particular, underlying Ilango's proof is a related decision problem about circuit complexity of Boolean \emph{simple extensions} which we call the $f$-Simple Extension Problem ($\SEP{f}$).
\begin{definition*} [Simple Extension]
    Let $f$ be a Boolean function that depends on all of its $n$ variables. A simple extension of $f$ is either $f$ itself or a function $g$ on $n + m$ variables satisfying:
    \begin{enumerate}
        \item $g$ depends on all of its inputs.
        \item $CC(g)$ --- the circuit-size complexity of $g$ --- is $CC(f) + m$.
        \item There exists a setting $k \in \{0,1\}^m$, a \emph{key}, such that for all $x \in \{0,1\}^n,$ $g(x, k) = f(x)$.
    \end{enumerate}
\end{definition*}
We define the $f$-Simple Extension decision problem for \emph{total functions} below.\footnote{For partial function $f$-Simple Extension ($\SEP{f}^*)$, $g$ is a \emph{partial} function and we must determine whether any completion of $g$ is a simple extension of $f$.} 
\begin{problem*} [The $f$-Simple Extension Problem]
    Let $f$ be a sequence of Boolean functions $\{f_n\}_{n \in \mathbb{N}}$ such that each $f_n$ depends on all of its $n$ inputs. The $f$-Simple Extension Problem is defined as follows: Given $n \in \mathbb{N}$ and $tt(g)$---the truth table of a binary function $g$---decide whether $g$ is a simple extension of $f_n$.
\end{problem*}

For a fixed $f$ whose truth table can be efficiently computed and whose exact circuit complexity is known, $\SEP{f}$ reduces to a single call to an $\MCSP$ oracle, because checking whether $g$ is a non-degenerate extension of $f$ can be done in polynomial time via brute force (given the truth table of $g$). This observation gives rise to an $\MCSP$-hardness proof framework: if one can identify an explicit function $f$ for which deciding the $f$-Simple Extension Problem is hard, then $\MCSP$ is also hard. 

This framework was implicitly used in Ilango's hardness proof for $\MCSP^*$ \cite{Ilango20}, i.e. reducing an ETH-hard problem to deciding whether a partial function is a simple extension of $f = \OR$. 
We make this observation explicit in Section \ref{sec:reframing}. 
Now, a natural question arises: can one extend this idea to total $\MCSP$ and prove $\MCSP \not\in \P$ assuming ETH? 
Ilango suggested that
\begin{quote}
    \emph{``the most promising approach is to skip $\MCSP^*$ entirely and extend our techniques to apply to $\MCSP$ directly.''} - Rahul Ilango, SIAM J. of Computing, 2022
\end{quote}

However, optimal circuits for $\OR$ are so well-structured that deciding whether a \emph{total} function is a simple extension of it is actually easy (see the discussion in Section 1.2.2 in \cite{Ilango20}). Minimal $\OR$ circuits are \emph{read-once formulas}: each of the input is read exactly once and each internal gate has fan-out $1$. Simple extensions of it will also be computed by read-once formulas, and deciding whether a given Boolean function has a read-once formula is \emph{easy} \cite{AngluinHK93, GMR06}. Therefore, 
\begin{quote}
    \emph{``the missing component in extending our results to $\MCSP$ is finding some function $f$ whose optimal circuits we can characterize but are also sufficiently complex.''} - Rahul Ilango, SIAM J. of Computing, 2022
\end{quote}

Could simply replacing $\OR$ with some read-many $f$ --- perhaps $\XOR$, which enjoys tight bounds and a full characterization --- allow Ilango's technique to prove $\MCSP$ is ETH-hard? For $\XOR$, we show the answer is a resounding \textbf{no.} For other potential functions, the answer is more ambiguous. We will discuss prospects for alternative hard functions in Section \ref{sec:future-discuss}.

\subsection{Our Results and Contributions}
\label{sec:our-results}
We narrow the field of candidate functions for such a hardness proof by developing a fixed-parameter tractable algorithm for the $f$-Simple Extension problem (Section \ref{sec:se-algorithm}). Such an algorithm is surprising, because $f$-Simple Extension is a meta-complexity problem about \emph{general circuits} and our algorithm works in regimes where we \emph{know} explicit circuit lower bounds. Often, the combinatorial facts used in lower bounds imply a hardness result for the appropriately-restricted meta-complexity problem (e.g., $\DNF$-$\MCSP$)!  Nonetheless, we obtain:

\begin{mainresult*}
    The $f$-Simple Extension Problem is in $\P$ whenever
    \begin{enumerate}
        \item $CC(f)$ --- the circuit-size complexity of $f$ --- is linear,
        \item the maximum fan-out over all optimal circuits for $f$ is constant, and
        \item the optimal\footnote{In $\cD$, we require a circuit to be \emph{normalized} for it to be optimal. In particular, it cannot contain any double-negations. This prevents every function from having an infinite number of optimal circuits.}circuits for $f$, up to isomorphism and permutation of its $n$ inputs, are efficiently enumerable and polynomial few with respect to the length of its truth table: $2^n$.
    \end{enumerate} 
\end{mainresult*}

To apply our main result and discount a particular $f$, we require an exact specification of its optimal circuits. The next natural candidate --- $\XOR$, a simple function whose circuits are neither read-once nor monotone --- has been well studied. Beyond enjoying exact size bounds \cite{Schnorr74,Redkin1973}, $\XOR$ is one of the few functions whose structure has been studied; it is known that, in $\cR$, \emph{all} optimal $\XOR_n$ circuits are binary trees of $n-1$ $\XOR_2$ sub-blocks \cite{Kombarov2011}. We extend this structural analysis to $\cD$ in Section \ref{sec:opt-XOR-ckt}, obtaining
\begin{mainlemma*}[\cite{Kombarov2011}, Theorem \ref{thm:XOR-structure}]
    Optimal $\XOR_n$ circuits consist of $(n-1)$ $(\neg)\XOR_2$ sub-circuits.
\end{mainlemma*}

Each $\XOR_n$ circuit can therefore be characterized using binary trees with $n$ leaves, of which there are $C_{n-1} = O(2^n)$, where $C_n$ is the $n^{th}$ Catalan number \cite{van_Lint_Wilson_1992}. As there are a finite number of optimal normalized $(\neg)\XOR_2$ circuits, combining this characterization and our main result to immediately yields
\begin{maincorollary*}
    The $\XOR$-Simple Extension Problem is in $\P$.
\end{maincorollary*}

Applying Ilango's technique successfully will itself require a deeper study of circuit minimization. This is not merely because any hardness proof needs to bypass our algorithm; knowledge of circuit lower-bounds and optimal constructions for the base function is intrinsic to the reduction itself. We make this connection explicit in Section \ref{sec:reframing}, identifying that

\begin{mainobservation*}
    $\SEP{f}^*$ is ETH-hard under Levin reductions.
\end{mainobservation*}

Lastly, in Section \ref{sec:roadmap}, we inspect each explicit function $f$ that enjoys DeMorgan circuit lower bounds and argue how plausible it is that the optimal set of $f$ circuits avoids our Main Result --- a \emph{roadmap} towards ETH-hardness of total $\MCSP$ via $\SEP{f}$.

\subsection{Related Work}

\subparagraph*{(Non-)hardness of $\MCSP$ Variants.} Hirahara showed that \emph{Partial $\MCSP$} is unconditionally $\NP$-hard under \emph{randomized reductions} \cite{Hirahara22}.  Extending his breakthrough result is another approach towards hardness of total $\MCSP$.  Though promising, this also faces challenges: under believable cryptographic conjectures (indistinguishability obfuscation and subexponentially-secure one-way functions), Gap$\MCSP$ is not $\NP$-complete under randomized Levin-reductions \cite{DBLP:conf/coco/MazorP24a}.  Such reductions appear to suffice for Hirahara's proofs, so one may need new ideas to obtain $\NP$-hardness of total $\MCSP$ via his approach.

We bypass the issue by working towards hardness of total $\MCSP$ under ETH --- a stronger assumption than $\P \neq \NP$.  Even so, ETH-hardness of total $\MCSP$ remains a major open problem, and there are no known barriers to extending Ilango's approach in this setting.  The reader can decide for themselves if our algorithm constitutes such a barrier or not.

\subsection{Discussion and Future Directions}
\label{sec:future-discuss}

\subparagraph{A Remark on Bases.}
Both $\cR$ and $\cD$ are compatible with Ilango's proof of ETH-hardness of $\MCSP^*$. That proof relied on functions whose optimal $\{\land, \lor, \neg\}$-circuits are read-once monotone formulas. There it was irrelevant whether $\neg$ gates contribute to size: those circuits simply did not contain negations. However, when we move to more complex functions like $\XOR$, negations \emph{must} appear in the optimal circuits and as such, we must decide how to treat them. 

In the search for non-linear circuit lower bounds, the choice between $\mu_R$ and $\mu_D$ is largely irrelevant: the two complexity measures the same up to a small constant factor. However, as we discuss in Section \ref{sec:reframing}, Ilango's technique requires knowledge of the \emph{structure} of optimal circuits. In this setting, there is not necessarily as strong connection between the two bases. By extending Kombarov's characterization of $\XOR$ circuits in $\cR$ to $\cD$ in Section \ref{sec:opt-XOR-ckt}, we show for that \emph{particular} function, optimal circuits are structurally similar in both bases. But, for other functions, this may well not be the case. It seems likely that negations could enable a function's optimal circuits to greatly vary under the two complexity measures. Indeed, negations can greatly increase a problems complexity: \cite{BlaisCOST2014} showed that a Boolean function learning problem became difficult only once the number of $\neg$ gates exceeded a small threshold.

By considering both bases in this work, we limit how negations impact the complexity of $\SEP{f}$. We show that they do not greatly increase the complexity of the simple extensions themselves: in Section \ref{sec:SE-tools}, we find that simple extensions under both complexity measures are highly structured. If a future hardness reduction relies on negations, their role will be in increasing the structural complexity of the underlying base function, either by increasing the number of distinct optimal circuits, or by enabling non-constant fan-out in a base circuit.

\subparagraph{A Roadmap for ETH-Hardness Proofs via Simple Extensions}
\label{sec:roadmap}
To prove $\SEP{f}$ is ETH-hard we will need a Boolean function whose optimal circuits are more complex and/or varied than $\XOR$. Specifically, these circuits must either (1) be superlinear in size, (2) require non-constant fanout, or (3) be sufficiently numerous. Superlinear bounds seem beyond current techniques---the most fruitful of which, gate elimination, seems unlikely to be able to prove lower bounds above even $11n$ for wide classes of functions \cite{GolovnevHKK18}. As such, it seems more sensible to identify Boolean functions which violate the latter conditions. However, structural characterization of the optimal circuits is also hard, and seems to require very tight circuit bounds. Indeed, our DeMorgan basis characterization of $\XOR$ repeatedly exploited Schnorr's \emph{exact} $3(n-1)$ bound for $\XOR_n$ \cite{Schnorr74}. Regardless, this greatly narrows the prospective class of functions from the original specification: ``more complex than read-once formulas.'' However, the known explicit functions with tight DeMorgan bounds that may violate these conditions are few and far between. In Table \ref{tab:candidates}\footnote{Some of listed bounds are in $\cU_2$, the basis consisting of every binary Boolean function besides $\XOR_2$ and $\neg\XOR_2$. For non-degenerate functions besides $f(x) = \neg x$, $\cU_2$ and $\cD$ are equivalent in terms of size. The multiplexer lower bound of \cite{Paul75} is for the $\cB_2$, the basis of all binary Boolean functions, but it also serves as the best known lower bound in $\cD$.}, we summarize these explicit functions and assess how suitable they are for ETH-hardness of $f$-Simple Extension Problem.

\begin{table}[t]
\centering
  \caption{Explicit Functions with Circuit Lower Bounds in the DeMorgan Basis}
  \label{tab:candidates}
  \begin{tabular}{ l c c c c c}
    Function(s) & Lower Bound  & & Upper Bound & $\Omega(1)$ Fanout & Source(s)\\
    \hline
    $\XOR$ & $\bm{3(n - 1)}$ & $\bm{=}$ & $\bm{3(n - 1)}$ &  NO & \cite{Schnorr74} \\
    Sum Mod 4 & $4n - O(1)$ & & $5n - O(1)$ & NO? & \cite{Zwick91} \\
    Sum Mod $2^k$ & $4n - O(1)$ & & $7n - o(n)$ & NO? & \cite{Zwick91} \\
    Multiplexer & $2(n - 1)$ & & $2n + O(\sqrt{n})$ & YES? & \cite{DBLP:journals/siamcomp/Paul77,KleinPaterson80} \\
    Well-Mixed &  $5n - o(n) $  & & $\poly$ & MAYBE? & \cite{LachishR01,IwamaM02} \\
    Weighted Sum of Parities &  $5n - o(n) $  & & $5n + o(n)$ & YES? & \cite{AmanoT11}
  \end{tabular}  
\end{table}

Observe that every function besides $\XOR$ in the table has a (small) gap between the circuit lower and upper bound, and the ``$\Omega(1)$ Fanout'' column ends with a question mark (?).  This is because $\XOR$ is the \emph{only} listed function for which we know the \emph{exact} circuit complexity and an optimal circuit characterization. The other ``$\Omega(1)$ Fanout'' entries above are extrapolated by assuming that their respective DeMorgan \emph{upper bound} constructions are optimal. For instance, Zwick conjectured that optimal circuits computing the Sum Mod 4 ($\MOD_4$) function are ``shaped like'' ternary full-adder blocks \cite{Zwick91}. If this conjecture is true, then $\MOD_4$-Simple Extension can be solved in $\poly$-time since such circuits satisfy the properties of our Main Lemma. Since the $\MOD_{2^k}$ functions are computed similarly, we conjecture that exactly characterizing the optimal circuits for Zwick's functions would yield efficient Simple Extension Solvers --- not a proof of ETH-hardness for total $\MCSP$.

However, we do have linear lower bounds for functions whose best known constructions have non-constant fanout: the multiplexing function ($\MUX$) contains sub-circuits which are reused a logarithmic number of times \cite{KleinPaterson80}.  In contrast to $\XOR$ however, the bounds for $\MUX$ are not tight.  The best lower bound is $2(n-1)$, given by Paul \cite{Paul75}.

\subparagraph{Future Directions.}
The most obvious next step is to either (1) obtain total characterization of the Multiplexer or (2) extend our Simple Extension Solver to handle circuits with super-constant fanout.  Neither of these tasks seems easy, but also they have not been subject to intensive research the way that super-linear circuit lower bounds and hardness of $\MCSP$ have.  We hope that connecting these kinds of results to ETH-hardness of $\MCSP$ provides new perspective and motivation.

Ilango's $\MCSP^*$ result has also formed the basis for several hardness results in other models of computation such as formulas and branching programs \cite{Ilango21, GlinskihR22, GlinskihR24}; could one show that the simple extension problem for formulas or branching programs is hard? These other models of computations may prove easier to work with than unrestricted circuits. They also enjoy superlinear lower bounds thereby bypassing our algorithm. Investigating the simple extension problem in these settings would provide insight into the feasibility of the approach in the circuit setting.

We conclude this section with a discussion of the simple extension problem in general. Despite having only been studied as a tool for proving hardness of $\MCSP$ thus far, it may be of independent interest.  For example, while hardness of time-bounded Kolmogorov complexity is tightly connected to the existence of one-way functions, hardness of $\MCSP$ has much weaker quantitative connections \cite{DBLP:conf/focs/LiuP20,RenS2021}.  The $f$-Simple Extension Problem is more ``structured'' than $\MCSP$ and easily reduces to it, so hardness assumptions about $f$-Simple Extension Problem are stronger. Could such assumptions imply one-way functions?

\subsection{Proof Techniques}
\subsubsection{The Structure of Optimal \texorpdfstring{$\XOR$}{XOR} Circuits}
\label{sec:xor-opt-ckt-summary}

Similar to \cite{Kombarov2011}, we show that \emph{every} optimal $(\neg)\XOR_n$ circuit over the DeMorgan basis partitions into trees of $(n-1)$ sub-circuits computing $(\neg)\XOR_2$  --- even when NOT gates are free. The structure of optimal circuits computing the $\XOR$-function is a crucial ingredient for ruling it out as a candidate function. 
We carry out an elementary but intricate case analysis of restricting and eliminating gates from optimal $\XOR$ circuits. Essentially we extract more information from the proof of Schnorr's lower bound by using it to identify ``templates'' that must be found in \emph{any} optimal $\XOR$ circuit.  We push this process to the limit, fully characterizing the ``shape'' of all such circuits. Specifically,
\begin{itemize}
    \item Schnorr's proof is essentially a technical lemma which says that any one-bit restriction will eliminate at least 3 costly gates \cite{Schnorr74}. This means that at the bottom level of every optimal $\XOR$ circuit, any variable must be fed into two distinct costly gates, and furthermore, one of these two must be fed into another costly gate. Any deviation from these properties will violate essential properties of the $\XOR$-function, such as ``$\XOR$ depends on all the input bits.'' Via a basic inductive argument and the fact that $\XOR$ is downward self-reducible, Schnorr's lower bound follows: $CC(\XOR_n) \geq 3(n-1)$.
    \item Schnorr's proof leaves the local structure of the optimal circuit computing $\XOR$ ``open.'' Namely, it does not provide any information about the other inputs of the costly gates or where their outputs connect to the rest of the circuit, since we consider fan-in 2 and unbounded fan-out. However, we know that $\XOR$ circuit has a matching upper-bound of $3(n-1)$. In particular, this means \emph{each one-bit restriction cannot remove more than $3$ gates}. We also know that \emph{each variable in optimal $\XOR$-circuits must be read twice}.
    \item We leverage these two properties to show that in every optimal $\XOR$ circuit, any two distinct input variables $x_i$ and $x_j$ must be fed into a block $\mathcal{B}$ as shown in the left sub-figure of Figure \ref{fig:XOR-true-shape}. Specifically, we argue that any deviations from the block will violate at least one of the properties via exhaustive case analyses of gate elimination steps. Finally, we argue that this block $\mathcal{B}$ must compute either $\XOR_2$ or $\lnot \XOR_2$ and apply a basic inductive argument to obtain the desired structural characterization of any optimal circuit computing $\XOR_n$ as depicted in the right sub-figure of Figure \ref{fig:XOR-true-shape}. 
\end{itemize}
Besides the linear size for optimal circuits computing $\XOR$, our structural theorem yields two more properties that rule out $\XOR$ as a candidate function for $\MCSP$-hardness via Simple Extension. That is, for optimal circuits computing $\XOR_n$, (1) \emph{the maximum fan-out is a constant}, and (2) \emph{the number of such optimal circuits up to permutation of variables, is $2^{O(n)}$}.

\begin{figure}[t]
    \centering
    \includegraphics[width=\textwidth]{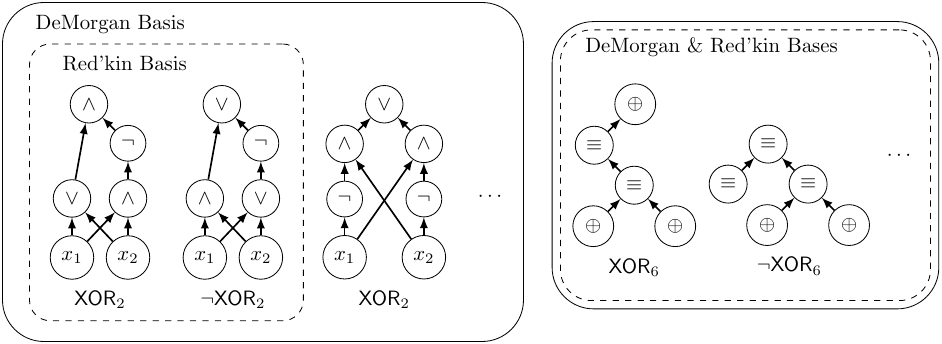}
    \caption{An example of the binary tree structure of optimal circuits computing $\XOR_6$. The left sub-figure depicts possible $(\neg)\XOR_2$ blocks in the Red'kin and DeMorgan Bases. Notice each optimal Red'kin circuit is an optimal DeMorgan circuit, but not vice-versa. The right sub-figure depicts that the arrangement of $\XOR_2$ blocks that make up $\XOR_6$ circuits are shared by both bases.}
    \label{fig:XOR-true-shape}
\end{figure}

\subsubsection{A Fixed-Parameter Tractable Simple Extension Solver}
\label{sec:fpt-sep-solver}
It is easy to see that the approach of solving the Simple Extension Problem for $\XOR_n$ via brute-forcing over all possible circuits of size $CC(f) + m$ is super-polynomial in terms of the length of the input truth-tables. Using the following ingredients, we design Algorithm \ref{alg:ckt-se-solver-informal} below, a Fixed-Parameter Tractable (FPT) algorithm for the Simple Extension problem that depends on the following three parameters: (1) the number of optimal circuits for $f$ (up to isomorphism \& permutation of variables), (2) the maximum fanout of any node in any optimal circuit for $f$, and (3) $CC(f)$.

\algnewcommand{\LComment}[1]{\State \textcolor{gray}{\(\triangleright\) #1}} 
\begin{algorithm}[t]
  \caption{Informal Simple Extension Solver, taking input
    $n \in \mathbb{N},~ g \in \cF_{n+m}$}
  \label{alg:ckt-se-solver-informal}
  \begin{algorithmic}[1]
    \If{there is no key to $f$ in $g$ or $g$ is degenerate}
    \State \Return \textbf{False}
    \EndIf
    \LComment{If the above tests pass, then $g$ is a non-degenerate extension of $f$.  It remains to check simplicity.}
    \For{each isomorphism class $\cC$ of open optimal circuits for $f$}
    \State $F \gets$ an arbitrary element of $\cC$ with all gates labelled in topological order
    \State{label the open nodes of $F$
      by an arbitrary permutation of $x_1, \dots , x_n$}
    \For{each reverse elimination $E$ that adds exactly $m$ costly gates to $F$}
    \State $\Tilde{G} \leftarrow \texttt{Decode}(F, E)$ 
    \If{$\texttt{tt}(\Tilde{G}) \simeq \texttt{tt}(g)$}
    \Comment{Test using the procedure of Theorem \ref{thm:tt-iso}.}
    \State \Return \textbf{True}
    \EndIf
    \EndFor
    \EndFor
    \State \Return \textbf{False}
  \end{algorithmic}
\end{algorithm}

\subparagraph*{Structured Simple Extension Circuits.}
By analyzing the \emph{behavior} of optimal simple extension circuits under gate elimination, we are able to characterize the \emph{structure} of \emph{every} optimal circuit computing a simple extension. By definition, if $g$ is a simple extension of $f$ then there are restrictions of $g$'s added variables (called \emph{extension} variables and denoted $y_i$) that yield $f$. We call such a restriction a \emph{key} to $f$ in $g$. We first show that \emph{circuits obtained by partially restricting with a key are themselves optimal simple extension circuits for intermediate extensions}. Building on this, we then develop convenient \emph{all-stops restrictions} that order substitutions and simplification steps with the following properties: (1) \emph{single-bit substitutions from this key in the given order eliminate exactly one costly gate at each step}, (2) \emph{there exists such an all-stops restriction for any optimal circuit computing a simple extension} (Lemma \ref{lem:exists-all-stops}).

Combining these tools, we inductively show a robust structure arises in optimal simple extension circuits: \emph{each extension variable occurs in an isolated read-once subformula that depends only on other extension variables} (referred to to as the \emph{Y-trees}). Formally,

\begin{definition}[Y-Tree Decomposition]
  Let $G$ be a circuit with two distinguished sets of inputs: \emph{base variables} $X$ and \emph{extension variables} $Y$.  A \emph{$Y$-Tree Decomposition} of $G$ is a set of triples $\langle \gamma , b , T \rangle$ where $\gamma$---referred to as a \emph{combiner}---is a costly gate of $G$, bit $b \in \bool$ designates an input of $\gamma$, and $T$ is a sub-circuit of $G$ rooted at the $b$ child of $\gamma$ such that
  
  \begin{enumerate}
  \item Each $T$ is a read-once formula in only extension variables $Y$.
  \item Each $y_i \in Y$ appears in at most one $T$.
  \item Each $T$ is \emph{isolated} in $G$ --- gate $\gamma$ is the unique gate reading from $T$, and it only reads the root of $T$.
  \item The sub-circuit of $G$ rooted at the $\neg$b child of $\gamma$ contains at least one $X$ variable.
  \end{enumerate}

  The \emph{size} of a decomposition is the number of tuples --- $Y$-trees and their associated combiner gates --- present in the decomposition. The \emph{weight} of a $Y$-tree decomposition is the number of extension variables that are read in some $T$. We say a $Y$-tree decomposition is \emph{total} if its weight is $|Y|$, i.e. every extension variable appears.  An example Y-tree decomposition of size three is depicted in Figure \ref{fig:example-y-tree}, where the shaded circles represent the circuity around each combiner $\gamma$ connecting each $T_\cY$ to the rest of the circuit.
\end{definition}

When gate elimination is performed with a total key, these added Y-trees and their combiners are pruned to reveal an embedded optimal circuit for the base $f$ function. We get the following structural insight: \emph{every optimal simple extension circuit has a total $Y$-tree decomposition}. This decomposition forms the basis of our strategy: we brute force over every optimal base circuit and try to ``splice in'' every possible $Y$-tree. 

\begin{figure}[t]
    \centering
    \includegraphics[]{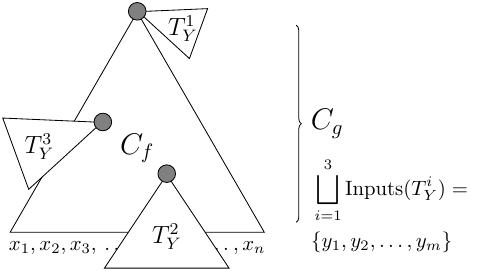}
    \caption{An example of a Y-Tree Decomposition of size three.}
    \label{fig:example-y-tree}
    
\end{figure}

\subparagraph*{Encoding \& Decoding the ``Grafts'' in a Y-Tree Decomposition.}
To ensure we can efficiently construct these candidate simple extension circuits, we devise an encoding scheme and corresponding decoding algorithm which efficiently captures the difference in local neighborhoods after each new Y-tree is spliced on top of an existing gate or input. Our final encoding must be $O(n + m)$ bits long to ensure brute-force runs in $2^{O(n+m)}$. We present our encoding as a communication problem to clarify the overhead and constraints involved.

Suppose $g$ is a simple extension of $f$ and Alice knows $G$, an optimal circuit for $g$.  Alice can obtain an optimal circuit $F$ computing $f$ by simply restricting the $y$-variables of $G$ with a key and performing gate elimination.  Now consider the following communication problem: Bob (i.e., line 5 of Algorithm \ref{alg:ckt-se-solver-informal}) knows $F$, and Alice would like to send him $G$ using as few bits as possible.  Because $g$ is a simple extension of $f$,  Alice can compute the $Y$-tree decomposition of optimal circuit $G$.  The idea is to send Bob a sequence of instructions that tell him exactly how to graft each Y-Tree of $G$ onto the gates of $F$, where all information is encoded relative to isomorphism-invariant properties of $F$.

\subparagraph*{Speeding Up Via Truth-table Isomorphism.}
Brute-forcing over the total encodings described above is still incredibly inefficient. Since each $Y$-tree is a read-once formula in the added $m$ variables there are least $C_{m-1} \cdot m!$ such explicit $Y$-trees, where $C_a$ is the $a^{th}$ Catalan number \cite{van_Lint_Wilson_1992}. The dominating term---$m!$---comes from permuting the labels of the variables. The same issue arises if our base function $f$ is symmetric: the number of optimal $f$ circuits is $\Omega(n!)$.

We sidestep this issue and drastically improve the speed of brute-force search. If we ``incorrectly'' assign variables in the base circuit or in the Y-trees, the result is a circuit for a Boolean function that is \emph{truth-table isomorphic} to $g$, i.e. their truth-tables are the same up to a permutation of the inputs. If $h$ and $g$ are truth-table isomorphic then they have the same circuit complexity. Thus it suffices to brute-force over unlabeled (``open'') base circuits and $Y$-trees, assign variables to inputs arbitrarily, generate each circuit's truth table, and check if it is truth-table isomorphic to $g$. This final step is feasible since truth-table isomorphism testing can be done in polynomial time \cite{Luks99}.

This results in our final algorithm, which runs in time $O(|L| \cdot 2^{O(\ell(s+m))})$ where $|L|$ is the number of optimal base $f$ circuits up to permutation of variables, $\ell$ is the maximum fanout in any of those base circuits, and $s = CC(f)$. As discussed above, for $\XOR$ these parameters are all sufficiently small and hence $\XOR$-Simple Extension is in $\P$.

\subsubsection{Paper Outline}
The paper is laid out as follows. Section \ref{sec:reframing} explores the implicit reduction present to $\SEP{f}$ in the ETH-hardness proof for $\MCSP^*$ in \cite{Ilango20}. This discussion yields our Main Observation and motivates our study of $\SEP{f}$ and circuit structures. It, in conjunction with this introduction, constitutes an extended abstract of our paper.

Appendix \ref{sec:Circuits} and Appendix \ref{sec:neo-gate-elimination} formally define circuits and gate elimination respectively. In Appendix \ref{sec:SE-tools}, we establish that every optimal circuit computing simple extensions are highly structured: they can be decomposed into $Y$-trees. This observation forms the basis of our Main Result, a fixed parameter tractable algorithm for $\SEP{f}$, in Appendix \ref{sec:xor-se-solver}. Finally, in Appendix \ref{sec:opt-XOR-ckt}, we extend Kombarov's characterization of optimal $\XOR$ circuits in $\cR$ to $\cD$, in order to apply our Main Theorem to $\SEP{\XOR}$. 

The dependence between sections varies. For example, Appendix \ref{sec:opt-XOR-ckt} is independent of the $\SEP{f}$ material like Section \ref{sec:reframing} and Appendix \ref{sec:SE-tools}. To aid the reader, we provide a reading order in Figure \ref{fig:reading-order}.
\begin{figure}[h]
    \centering
    \includegraphics[width=\textwidth]{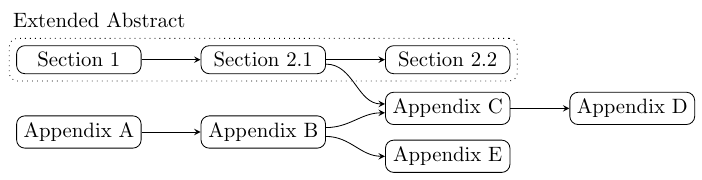}
    \caption{The structure of our paper. Sections 1 and 2 form an extended abstract. An incoming arrow indicates that the material depends on the previous section.}
    \label{fig:reading-order}
\end{figure}

\section{Revisiting ETH hardness for \texorpdfstring{$\MCSP^*$}{MCSP*} via Simple Extensions} \label{sec:reframing}
We re-examine the proof that $\MCSP^*$ is ETH-hard from \cite{Ilango20}. Our aim is to better understand the reduction from $2n \times 2n$ Bipartite Permutation Independent Set ($\BPIS$) to the Partial $f$ Simple Extension problem ($\SEP{f}^*)$. 

\subsection{The \texorpdfstring{$f$}{f}-Simple Extension Problem}

We give formal definitions of simple extensions and its associated decision problem. We first define non-degeneracy of a Boolean function. 

\begin{definition}
\label{def:bit-dependency}
    A function $f \in \cF_n$, the set of Boolean functions on $n$ variables, depends on its $i^{th}$ variable, $x_i$, if there exists an input $\alpha \in \cF_n$ such that
    $ f(\alpha) \neq f(\alpha \oplus e_i), $
    where $e_i$ denotes the Boolean vector that is $0$ everywhere except for a $1$ at index $i$ and $\oplus$ is bitwise $\XOR$. If $f$ depends on all of its variables, then we say $f$ is a \emph{non-degenerate} function. Conversely, we say $f$ is a \emph{degenerate} function if it does not depend on at least one variable.
\end{definition}

We now define simple extension as
\begin{definition}
\label{def:simp-ext}
    Let $f \in \cF_n$ be non-degenerate. A \emph{simple extension} of $f$ is either $f$ itself or a function $g \in \cF_{n+m}$ satisfying:
    \begin{enumerate}
        \item $g$ is a non-degenerate function,
        \item $CC(g) = CC(f) + m$, and
        \item there exists a setting $k \in \{0,1\}^m$, called a \emph{key}, such that for all $x \in \{0,1\}^n,$ $g(x, k) = f(x)$.
    \end{enumerate}
\end{definition}

We denote the first $n$ inputs of $f$ and $g$ by $x_1, \ldots x_n$ and will refer to the extra $m$ inputs of $g$ as \emph{extension variables} and refer to them as $y_1, \ldots, y_m$. From the definition of simple extension above, we define the following decision problem.

\begin{problem}[$\SEP{f}$]
\label{problem:decide-m-SE}
    Let $f$ be a sequence of Boolean functions $\{f_n\}_{n \in \mathbb{N}}$ such that each $f_n$ is a non-degenerate function in $\cF_n$. The $f$-Simple Extension Problem is defined as follows: Given $n \in \mathbb{N}$ and $tt(g)$---the truth tables of a binary function $g \in \cF_{n+m}$---decide whether $g$ is a simple extension of $f_n$.
\end{problem}

We extend this to partial functions,

\begin{problem}[$\SEP{f}^*$]
\label{problem:decide-m-SE*}
    Given $n \in \mathbb{N}$ and a \emph{partial} $tt(g)$, decide whether any completion of the truth table is a simple extension of $f_n$.
\end{problem}

\subsection{An Explicit Reduction \texorpdfstring{$\BPIS$}{BPIS} from to \texorpdfstring{$\SEP{f}^*$}{f-SEP*}}

Recall the formulation of $\BPIS$ from \cite{Ilango20},
\begin{problem}[$\BPIS$]
    The $2n \times 2n$ Bipartite Permutation Independent Set is defined as follows: Given $G = (V, E)$ a directed graph with vertex set $V = [n] \times [n]$. Decide whether there exists a permutation $\pi: [2n] \to [2n]$ such that:
        \begin{enumerate}
            \item $\pi([n]) = [n]$,
            \item $\pi(\{n + i : i \in [n]\}) = \{n + i:i \in [n]\}$,
            \item if $((j, k), (j', k')) \in E$, then either $\pi(j) \neq k$ or $\pi(j' + n) \neq k' + n$
        \end{enumerate}
\end{problem}

If the Exponential Time Hypothesis holds, then $\BPIS$ cannot be solved much faster than brute-forcing over all $2^{o(n \log n)}$ permutations \cite{Lokshtanov18}. $\BPIS$ is a natural problem to show hardness of circuit size problems since there are $O(2^{s \log s})$ circuits of size $s$. Reducing from $\BPIS$ implies, assuming ETH, that our problem can not be solved much faster than by brute-forcing over all possible circuits.

The reduction from $\BPIS$ to $\SEP{f}^*$, as it appears as part of the original proof, is somewhat \emph{implicit}; the target problem, as written, could be described more simply as ``determine whether a partial truth table is ever consistent with a monotone read-once formula.'' Since this is \emph{easy} for total functions \cite{AngluinHK93, GMR06}, this view of the reduction cannot help us to extend the reduction technique to total $\MCSP$. We will need the more general $\SEP{f}$ and by reframing Ilango's proof we gain insight into how it might extend to total functions.

The connection to $\SEP{f}^*$ is \emph{explicitly} stated, however, in the introduction of \cite{Ilango20}. Here, the reduction is identified with $\SEP{\OR_{4n}}^*$ where each $z$ variable is an extension variable. This is true, though non-degenerate functions computed by read-once formulas are simple extensions of \emph{any} of their non-degenerate restrictions. When framing the reduction to $\SEP{f}^*$ explicitly, we find it more compelling to choose a different function $\hat{f}$, whose truth table is given by $\bigvee_{i \in [2n]}(y_i \land z_i)$ --- because using $\hat{f}$ makes the intuitive description of Ilango's technique ``reverse gate elimination'' an obvious property of the reduction. Under this framing, the extension variables will instead be the $x$ variables in Ilango's original proof. Despite this, fixing $f$ to be $\OR_{4n}$ is \emph{not} arbitrary; afterwards, we discuss what insight it provides.  Informally, the two choices of base function provide distinct ``channels'' for ETH to imply hardness of $\SEP{f}^*$.

\subparagraph*{Structural Lemmas} To encode $\BPIS$ in $\SEP{\hat{f}}^*$, we first prove two lemmas which were essentially proven in tandem in \cite{Ilango20} as Lemma 16. We separate them out here and stay faithful to the original arguments. Like Lemma 16, these lemmas establish structural properties of circuits computing our base function $\hat{f}$ and our eventual output $\hat{g}$. 
\begin{lemma}\label{lem:rahul-base-lemma}
    Let $\hat{f} : \{0,1\}^{2n} \times \{0,1\}^{2n} \to \{0,1\}$ be the Boolean function computed by $\bigvee_{i \in [2n]}(y_i \land z_i)$. If $\psi$ be an optimal \textit{normalized}\footnote{A formula is normalized if all negations are pushed down to the input level. Normalization does not affect the size of the formula, and thus $\SEP{f}^*$ still reduces to $\MCSP^*$ even if we restrict ourself to normalized formulas.} formula computing $\hat{f}$, then $\psi$, as a formula, is equal to $\bigvee_{i \in [2n]}(y_i \land z_i)$. 
\end{lemma}
In particular, the only difference between two distinct optimal circuits for $\hat{f}$ is which binary tree of fanin 2 $\lor$ gates is used.

\begin{proof}
    Observe that $\psi$ must read all of its input variables and furthermore must do so positively. This is because (1) $f$ depends on all of its variables and is monotone in them, and (2) $\psi$ is normalized and thus no $\neg$ gates can appear internally in the circuit. We note that $\psi$ must use a total of $4n-1$ gates since $f$ is non-degenerate on $4n$ variables computable. We see that $\psi$ must contain at least $2n-1$ $\lor$ gates: substituting $z = 1^{2n}$ and simplifying yields a monotone read once formula computing $\OR_{2n}(y_1, \ldots, y_{2n})$ which must consist of $2n-1$ $\lor$ gates that appear in $\psi$. 
    
    We now argue that each $z_i$ feeds into an $\land$ gate. Assume otherwise, then observe that setting $z_i = 1$ and simplifying removes at least two gates (since $z_1 \lor p\equiv 1 \lor p \equiv 1)$. The resulting read once formula for $\hat{f}|_{z_i \leftarrow 1}$ uses at most $4n-3$ gates. This is a contradiction since $\hat{f}|_{z_i \leftarrow 1}$ still depends on all of its remaining $4n-1$ variables: it cannot be computed by a circuit with fewer than $4n-2$ gates.

    We now argue the other input to the $\land$ gate fed by $z_i$ is $y_i$. Assume otherwise. Notice that since $\psi$ is read-once, setting $z_i \leftarrow 0$ simplifying disconnects the other input to $\land$ gate and thus removes dependence on any variables it depends on. However, $\hat{f}|_{z_i \leftarrow 0}$ still depends on all of its variables besides $y_i$. Thus the $\land$ gate can only read $y_i$.

    Since each $z_i$ and $y_i$ feed a distinct $\land$ gate, there are at least $2n$ $\land$ gates. Since there are $4n-1$ total gates, and at least $2n-1$ $\lor$ gates, we know that these $\land$ gates are the only $\land$ gates that appear. Thus the remainder of the circuit is a binary tree of $2n-1$ $\lor$ gates whose $2n$ leaves are the $2n$ $\land$ gates. 
\end{proof}

Having established the structure of optimal circuits computing $\hat{f}$ we can further restrict its simple extensions. Extra restrictions to the truth table enforce that extension variables must be spliced into the circuit using $2n$ additional $\lor$ gates that each read a different $y_i$. This pairing of each $y_i$ with a different $x_j$ will define a permutation in which we can encode $\BPIS$ solutions.

\begin{lemma}\label{lem:rahul-structural-lemma}
    Let $g : \{0,1\}^{2n} \times \{0,1\}^{2n} \times \{0,1\}^{2n} \to \{0,1\}$ be any simple extension of $\hat{f}$ satisfying the following conditions:
    \begin{equation*}
        g(x,y,z) = 
            \begin{cases}
                \hat{f}(y,z) & \text{if } x = 0^{2n} \\
                \OR_{2n}(z_1, \ldots z_{2n}) & \text {if } x = 1^{2n} \\
                \OR_{4n}(x_1,\ldots x_{2n}, y_1, \ldots y_{2n}) & \text {if } z = 1^{2n} \\
                0 & \text {if } z = 0^{2n}
            \end{cases}
    \end{equation*}
    If $\phi$ is an optimal normalized formula computing $g$ then there exists a permutation $\pi : [2n] \to [2n]$ such that $\phi$ equals, as a formula, $\bigvee_{i \in [2n]}((x_{\pi(i)} \lor y_i) \land z_i)$.
\end{lemma}

\begin{proof}
    Since $g(x,y,z) = \hat{f}(y,z)$ when $x = 0^{2n}$, we know that $\phi$ must read all $y$ and $z$ variables positively. Similarly, all $x$ variables must be read positively, since $g(x,y,1^{2n}) = \OR_{4n}(x_1,\ldots x_{2n}, y_1, \ldots y_{2n})$. Note that these restrictions also imply that $\phi$ contains exactly $4n - 1$ $\lor$ gates and $2n$ $\land$ gates since substituting and simplifying yields optimal formulas for those restrictions. From the lemma above, we know know that when setting $x = 0^{2n}$ and simplifying, we obtain a circuit structurally equivalent to $\bigvee_{i \in [2n]}(y_i \land z_i)$. Therefore $2n$ $\lor$ gates must be removed during simplification. These $\lor$ gates cannot feed any remaining $\lor$ above the $\land$ gates, since otherwise setting $x = 1^{2n}$ would fix the circuit to be $1$, rather than $\OR_{2n}(z_1, \ldots z_{2n})$. Similarly, $z_i$ cannot feed any of these $\lor$ gates, as setting $x = 1^{2n}$ would remove dependence on that $z_i$ since the circuit is a read-once formula. Thus each $\lor$ gate can only depend on the $x$ and $y$ variables. Observe that each $y_i$ must feed into one of these $\lor$ gates instead of the $\land$ gate fed by $z_i$, as otherwise when we set $x = 1^{2n}$, the function would still depend on $y_i$. Since there are exactly $2n$ additional $\lor$ gates, and exactly $2n$ $y$ and $2n$ $x$ variables, its easy to see that each additional $\lor$ gate must read one $x_i$ and one $y_j$. Therefore, as a formula, $\phi$ must be $\bigvee_{i \in [2n]}((x_{\pi(i) \lor y_i}) \land z_i)$ for some permutation $\pi$.
\end{proof}

\begin{table}[t]
\centering
\caption{$\BPIS$ requirements and the corresponding restrictions on $\hat{g}$}
\label{tab:BPIS-to-SEP}
\renewcommand\tabularxcolumn[1]{m{#1}}
\renewcommand{\arraystretch}{1.2}
\begin{tabularx}{\textwidth}{l|l|X}
$\BPIS$ Requirement on $\sigma$ & Corresponding $\hat{g}$ Restriction & Impact If $\pi$ Violates  \\ \hline
$\sigma(\{1, \ldots, n\}) = \{1, \ldots, n\} $& \pbox{\textwidth}{$\OR_n(x_1, \ldots, x_n)$ when \\ $z = 1^n0^n$ and $y = 0^{2n}$} &  If $\pi(i) = j \geq n$, then $z_j \leftarrow 0$ removes $x_i$ in $C_\pi$\\ \hline
$\begin{aligned} \sigma(\{&n+1, \ldots, 2n\})\\ &= \{n+1, \ldots, 2n\}\end{aligned}$& \pbox{\textwidth}{$\OR_n(x_{n+1}, \ldots x_{2n})$ when \\ $z = 0^n1^n$ and $y = 0^{2n}$}& As above, $C_\pi\hook_{z_j \leftarrow 0}$ will not depend on $x_i$ \\ \hline
\pbox{\textwidth}{If $((j, k), (j',k')) \in E$ then \\ $\sigma(j) \neq k$ or $\sigma(n+j') \neq \sigma(n + k')$}& \pbox{\textwidth}{$1$ if $(x, y, z) = (\overline{e_k}\overline{e_{k'}}, 0^{2n}, e_je_{j'})$ \\ where $\exists((j,k),(j',k')) \in E$} & $C_\pi$ wrongly outputs $0$
\end{tabularx}
\end{table}

\subparagraph*{An Explicit Reduction} We now provide an explicit reduction from $\BPIS$ to $\SEP{\hat{f}}^*$. Given an instance $G$ of $\BPIS$ we output $4n$ and the partial truth-table for a function $\hat{g}$ that is consistent with the requirements of Lemma \ref{lem:rahul-structural-lemma}. We add three additional restrictions (listed in Table \ref{tab:BPIS-to-SEP}) to ensure that any permutation $\pi$, whose corresponding circuit $C_\pi \equiv \bigvee_{i\in[2n]}\left((x_{\pi(i)}\lor y_i)\land z_i\right)$ is consistent with $\hat{g}$, \emph{is} also a solution for $G$ (and vice versa). All other rows of the truth table are left undefined (e.g. as $\star$). We summarize the requirements for any valid $\BPIS$ solution $\sigma$, the corresponding restriction, and how it enforces the requirement on $\pi$ in Table \ref{tab:BPIS-to-SEP}.

This completes the reduction as any $\sigma$ satisfying $\BPIS$ for $G$ can be used to construct a read-once circuit consistent with $\hat{g}$ and vice versa. The arguments verifying this are the same as in \cite{Ilango20}, and we refer the reader there for the full details.

\begin{lemma}
    $\BPIS$ reduces to $\SEP{\hat{f}}^*$ in $2^{O(n)}$ time.
\end{lemma}

\subparagraph*{The Original Framing} In its introduction, \cite{Ilango20} identifies $
\hat{g}$ as a simple extension of $\OR_{4n}$. Under this lens, the hardness comes from determining which $\OR_{4n}$ base circuit can have the $z$ extension variables added. The additional truth-table restrictions on $\hat{g}$ force each $z_i$ to be spliced in a particular way adjacent to $y_i$. Assuming ETH, there are $\Omega(n!)$ optimal base $\OR_{4n}$ circuits that must be checked via brute-force.

\subparagraph*{Implicit Circuit Lower Bounds \& Enumeration of Optimal Circuits} From both presentations, we see that leveraging $\SEP{f}$ involves explicit circuit size lower bounds. Indeed, both Lemma \ref{lem:rahul-base-lemma} and Lemma \ref{lem:rahul-structural-lemma} prove formula lower bounds for specific non-degenerate functions. However, in a sense, circuit lower bounds are intrinsic to the reduction itself. This connection can be made rigorous: the reduction can used to produce explicit Boolean functions which enjoy non-vacuous lower bounds. On \texttt{no} instance of $\BPIS$, the reduction outputs a partial truth table where every completion is non-degenerate but \emph{not} a simple extension. Hence, the circuit complexity of these completions is not the vacuous $6n-1$ lower bound obtained by knowing that functions produced are non-degenerate.

Furthermore, the reduction did not solely rely on the circuit complexity of $\hat{f}$ and $\hat{g}$. Lemmas \ref{lem:rahul-base-lemma} and \ref{lem:rahul-structural-lemma} \emph{tightly control} how base circuits and their extensions can be arranged; and this is pivotal for encoding $\BPIS$ permutations. This structural requirement can be formalized by observing the reduction is also an efficient \emph{Levin reduction}.

Recall, from \cite{DBLP:conf/coco/MazorP24a}, that a Levin reduction is a many-one reduction that also efficiently maps witnesses, not just problem instances. More precisely, let $R$ be a set of ordered pairs $(x,w)$ where $x$ is a yes-instance of a problem and $w$ is an accompanying certificate. We define $L_R$, the language defined by $R$, to be the set of elements $x$ such that $(x,w) \in R$ for some $w$. Then a Levin reduction between two languages $A$ and $B$ is an efficient many-one reduction $r$ between problem instances paired with two efficient mappings $m, \ell$ between instance-witness pairs that satisfy (1) if $(x,w) \in R_A$ then $(r(x),m(x,w)) \in R_B$ and (2) $(t(x), w) \in R_B$ implies $(x, \ell(x,w)) \in R_A$. 

For $\BPIS$, the witnesses for an instance are simply the valid permutations $\sigma$ and witnesses for $\SEP{\hat{f}}$ are optimal circuits computing the extension. Let $\cR_{\BPIS}$ and $\cR_{\SEP{\hat{f}}^*}$ be the sets of ordered pairs consisting of problem instances and all of their witnesses as described. The reduction admits linear time mappings between witnesses: given $\sigma$, simply construct $\bigvee_{i \in [2n]}((x_{\sigma(i)} \lor y_i) \land z_i)$ and given a circuit for $\hat{g}$, simply read off the permutation from the $x$ variables.
\clearpage

\bibliography{references}

@ARTICLE{Shannon49,
    author={Shannon, Claude. E.},
    journal={The Bell System Technical Journal}, 
    title={The synthesis of two-terminal switching circuits}, 
    year={1949},
    volume={28},
    number={1},
    pages={59-98},
    doi={10.1002/j.1538-7305.1949.tb03624.x}
}

@article{Masek1979,
    title={Some NP-complete set covering problems},
    author={Masek, William J},
    journal={Unpublished manuscript},
    year={1979}
}

@ARTICLE{Trakhtenbrot84,  
    author={B. A. {Trakhtenbrot}},  
    journal={Annals of the History of Computing},   
    title={A Survey of Russian Approaches to Perebor (Brute-Force Searches) Algorithms},   
    year={1984},  
    volume={6},  
    number={4},  
    pages={384-400},  
    doi={10.1109/MAHC.1984.10036}
}

@article{Redkin1973,
    title={Proof of minimality of circuits consisting of functional elements},
    author={Red’kin, NP},
    journal={Systems Theory Research: Problemy Kibernetiki},
    pages={85--103},
    year={1973},
    publisher={Springer}
}

@article{Schnorr74,
    author    = {Claus{-}Peter Schnorr},
    title     = {Zwei lineare untere Schranken f{\"{u}}r die Komplexit{\"{a}}t
               Boolescher Funktionen},
    journal   = {Computing},
    volume    = {13},
    number    = {2},
    pages     = {155--171},
    year      = {1974},
    url       = {https://doi.org/10.1007/BF02246615},
    doi       = {10.1007/BF02246615},
    bibsource = {dblp computer science bibliography, https://dblp.org}
}

@book{Wegener1987,
    place     = {Stuttgart},
    title     = {The Complexity of Boolean Functions},
    publisher = {Wiley-Teubner},
    author    = {Wegener, Ingo},
    year      = {1987}
}

@article{Zwick91,
  title={A 4n lower bound on the combinational complexity of certain symmetric boolean functions over the basis of unate dyadic Boolean functions},
  author={Zwick, Uri},
  journal={SIAM Journal on Computing},
  volume={20},
  number={3},
  pages={499--505},
  year={1991},
  publisher={SIAM}
}

@article{AngluinHK93,
    author    = {Dana Angluin and
               Lisa Hellerstein and
               Marek Karpinski},
    title     = {Learning Read-Once Formulas with Queries},
    journal   = {J. {ACM}},
    volume    = {40},
    number    = {1},
    pages     = {185--210},
    year      = {1993},
    url       = {https://doi.org/10.1145/138027.138061},
    doi       = {10.1145/138027.138061},
    bibsource = {dblp computer science bibliography, https://dblp.org}
}

@inproceedings{Luks99,
    title={Hypergraph isomorphism and structural equivalence of boolean functions},
    author={Luks, Eugene M},
    booktitle={Proceedings of the thirty-first annual ACM symposium on Theory of computing},
    pages={652--658},
    year={1999}
}

@inproceedings{KabanetsC00,
    author    = {Valentine Kabanets and
               Jin{-}yi Cai},
    editor    = {F. Frances Yao and
               Eugene M. Luks},
    title     = {Circuit minimization problem},
    booktitle = {Proceedings of the Thirty-Second Annual {ACM} Symposium on Theory
               of Computing, May 21-23, 2000, Portland, OR, {USA}},
    pages     = {73--79},
    publisher = {{ACM}},
    year      = {2000},
    url       = {https://doi.org/10.1145/335305.335314},
    doi       = {10.1145/335305.335314},
    bibsource = {dblp computer science bibliography, https://dblp.org}
}

@inproceedings{LachishR01,
    author = {Lachish, Oded and Raz, Ran},
    title = {Explicit lower bound of 4.5n - o(n) for boolena circuits},
    year = {2001},
    isbn = {1581133499},
    publisher = {Association for Computing Machinery},
    address = {New York, NY, USA},
    url = {https://doi.org/10.1145/380752.380832},
    doi = {10.1145/380752.380832},
    booktitle = {Proceedings of the Thirty-Third Annual ACM Symposium on Theory of Computing},
    pages = {399–408},
    numpages = {10},
    location = {Hersonissos, Greece},
    series = {STOC '01}
}

@inproceedings{IwamaM02,
  author       = {Kazuo Iwama and
                  Hiroki Morizumi},
  editor       = {Krzysztof Diks and
                  Wojciech Rytter},
  title        = {An Explicit Lower Bound of 5n - o(n) for Boolean Circuits},
  booktitle    = {Mathematical Foundations of Computer Science 2002, 27th International
                  Symposium, {MFCS} 2002, Warsaw, Poland, August 26-30, 2002, Proceedings},
  series       = {Lecture Notes in Computer Science},
  volume       = {2420},
  pages        = {353--364},
  publisher    = {Springer},
  year         = {2002},
  url          = {https://doi.org/10.1007/3-540-45687-2\_29},
  doi          = {10.1007/3-540-45687-2\_29},
  bibsource    = {dblp computer science bibliography, https://dblp.org}
}

@article{GMR06,
    title = {Factoring and recognition of read-once functions using cographs and normality and the readability of functions associated with partial k-trees},
    journal = {Discrete Applied Mathematics},
    volume = {154},
    number = {10},
    pages = {1465-1477},
    year = {2006},
    issn = {0166-218X},
    doi = {https://doi.org/10.1016/j.dam.2005.09.016},
    url = {https://www.sciencedirect.com/science/article/pii/S0166218X06000072},
    author = {Martin Charles Golumbic and Aviad Mintz and Udi Rotics}
}

@article{Kombarov2011,
    title={The minimal circuits for linear Boolean functions},
    author={Kombarov, Yu A},
    journal={Moscow University Mathematics Bulletin},
    volume={66},
    number={6},
    pages={260--263},
    year={2011},
    publisher={Springer}
}

@article{AmanoT11,
  title={A well-mixed function with circuit complexity 5n: Tightness of the Lachish--Raz-type bounds},
  author={Amano, Kazuyuki and Tarui, Jun},
  journal={Theoretical computer science},
  volume={412},
  number={18},
  pages={1646--1651},
  year={2011},
  publisher={Elsevier}
}

@book{Jukna2012,
  author       = {Stasys Jukna},
  title        = {Boolean Function Complexity - Advances and Frontiers},
  series       = {Algorithms and combinatorics},
  volume       = {27},
  publisher    = {Springer},
  year         = {2012},
  url          = {https://doi.org/10.1007/978-3-642-24508-4},
  doi          = {10.1007/978-3-642-24508-4},
  isbn         = {978-3-642-24507-7},
  bibsource    = {dblp computer science bibliography, https://dblp.org}
}

@article{BlaisCOST2014,
  title={Learning circuits with few negations},
  author={Blais, Eric and Canonne, Cl{\'e}ment L and Oliveira, Igor C and Servedio, Rocco A and Tan, Li-Yang},
  journal={arXiv preprint arXiv:1410.8420},
  year={2014}
}

@inproceedings{HiraharaS2017,
    title={On the average-case complexity of MCSP and its variants},
    author={Hirahara, Shuichi and Santhanam, Rahul},
    booktitle={32nd Computational Complexity Conference (CCC 2017)},
    year={2017},
    organization={Schloss Dagstuhl-Leibniz-Zentrum fuer Informatik}
}

@inproceedings{HiraharaOS18,
  author       = {Shuichi Hirahara and
                  Igor C. Oliveira and
                  Rahul Santhanam},
  editor       = {Rocco A. Servedio},
  title        = {NP-hardness of Minimum Circuit Size Problem for {OR-AND-MOD} Circuits},
  booktitle    = {33rd Computational Complexity Conference, {CCC} 2018, June 22-24,
                  2018, San Diego, CA, {USA}},
  series       = {LIPIcs},
  volume       = {102},
  pages        = {5:1--5:31},
  publisher    = {Schloss Dagstuhl - Leibniz-Zentrum f{\"{u}}r Informatik},
  year         = {2018},
  url          = {https://doi.org/10.4230/LIPIcs.CCC.2018.5},
  doi          = {10.4230/LIPICS.CCC.2018.5},
  bibsource    = {dblp computer science bibliography, https://dblp.org}
}

@article{Lokshtanov18,
    author = {Lokshtanov, Daniel and Marx, D\'{a}niel and Saurabh, Saket},
    title = {Slightly Superexponential Parameterized Problems},
    journal = {SIAM Journal on Computing},
    volume = {47},
    number = {3},
    pages = {675-702},
    year = {2018},
    doi = {10.1137/16M1104834},
    
    URL = {  
            https://doi.org/10.1137/16M1104834
        },
    eprint = {    
            https://doi.org/10.1137/16M1104834
        }
}

@article{GolovnevHKK18,
    author    = {Alexander Golovnev and
               Edward A. Hirsch and
               Alexander Knop and
               Alexander S. Kulikov},
    title     = {On the limits of gate elimination},
    journal   = {J. Comput. Syst. Sci.},
    volume    = {96},
    pages     = {107--119},
    year      = {2018},
    url       = {https://doi.org/10.1016/j.jcss.2018.04.005},
    doi       = {10.1016/j.jcss.2018.04.005},
    bibsource = {dblp computer science bibliography, https://dblp.org}
}

@inproceedings{GolovnevIIKKT2019,
    title={AC0 [p] lower bounds against MCSP via the coin problem},
    author={Golovnev, Alexander and Ilango, Rahul and Impagliazzo, Russell and Kabanets, Valentine and Kolokolova, Antonina and Tal, Avishay},
    booktitle={ICALP},
    year={2019}
}

@article{Ilango20,
    author = {Ilango, Rahul},
    title = {Constant Depth Formula and Partial Function Versions of MCSP Are Hard},
    journal = {SIAM Journal on Computing},
    volume = {0},
    number = {0},
    pages = {FOCS20-317-FOCS20-367},
    year = {2020},
    doi = {10.1137/20M1383562},
    URL = {    
            https://doi.org/10.1137/20M1383562
        },
    eprint = {    
            https://doi.org/10.1137/20M1383562
        }
}

@inproceedings{IlangoLO20,
    title={NP-Hardness of Circuit Minimization for Multi-Output Functions.},
    author={Ilango, Rahul and Loff, Bruno and Oliveira, Igor Carboni},
    booktitle={Electronic Colloquium on Computational Complexity (ECCC)},
    volume={27},
    pages={21},
    year={2020}
}

@article{Santhanam2020,
    title={Pseudorandomness and the minimum circuit size problem},
    author={Santhanam, Rahul},
    journal={LIPIcs},
    volume={151},
    year={2020},
    publisher={Schloss Dahgstuhl}
}

@article{CheraghchiKLM2020,
    title={Circuit lower bounds for MCSP from local pseudorandom generators},
    author={Cheraghchi, Mahdi and Kabanets, Valentine and Lu, Zhenjian and Myrisiotis, Dimitrios},
    journal={ACM Transactions on Computation Theory (TOCT)},
    volume={12},
    number={3},
    pages={1--27},
    year={2020},
    publisher={ACM New York, NY, USA}
}

@article{RenS2021,
    title={Hardness of KT characterizes parallel cryptography},
    author={Ren, Hanlin and Santhanam, Rahul},
    journal={Cryptology ePrint Archive},
    year={2021}
}

@inproceedings{Hirahara22,
    author       = {Shuichi Hirahara},
    title        = {NP-Hardness of Learning Programs and Partial {MCSP}},
    booktitle    = {63rd {IEEE} Annual Symposium on Foundations of Computer Science, {FOCS}
                  2022, Denver, CO, USA, October 31 - November 3, 2022},
    pages        = {968--979},
    publisher    = {{IEEE}},
    year         = {2022},
    url          = {https://doi.org/10.1109/FOCS54457.2022.00095},
    doi          = {10.1109/FOCS54457.2022.00095},
    bibsource    = {dblp computer science bibliography, https://dblp.org}
}

@article{HuangIR2023,
    title={NP-Hardness of Approximating Meta-Complexity: A Cryptographic Approach},
    author={Huang, Yizhi and Ilango, Rahul and Ren, Hanlin},
    journal={Cryptology ePrint Archive},
    year={2023}
}

@inbook{van_Lint_Wilson_1992, 
    place={Cambridge}, 
    title={Recursions and generating functions}, 
    booktitle={A Course in Combinatorics}, 
    publisher={Cambridge University Press}, 
    author={van Lint, J. H. and Wilson, R. M.}, 
    year={1992}, pages={109–1311}
}

@misc{Feller1968,
  title={An Introduction to Probability Theory and its Applications. Vol I},
  author={Feller, William and Morse, Philip M},
  edition={Third},
  year={1968},
  publisher={American Institute of Physics},
  pages={38,51}
}

@article{ArvindV14,
  author       = {Vikraman Arvind and
                  Yadu Vasudev},
  title        = {Isomorphism testing of Boolean functions computable by constant-depth
                  circuits},
  journal      = {Inf. Comput.},
  volume       = {239},
  pages        = {3--12},
  year         = {2014},
  url          = {https://doi.org/10.1016/j.ic.2014.08.003},
  doi          = {10.1016/J.IC.2014.08.003},
  bibsource    = {dblp computer science bibliography, https://dblp.org}
}

@inproceedings{Ilango21,
  author       = {Rahul Ilango},
  title        = {The Minimum Formula Size Problem is {(ETH)} Hard},
  booktitle    = {62nd {IEEE} Annual Symposium on Foundations of Computer Science, {FOCS}
                  2021, Denver, CO, USA, February 7-10, 2022},
  pages        = {427--432},
  publisher    = {{IEEE}},
  year         = {2021},
  url          = {https://doi.org/10.1109/FOCS52979.2021.00050},
  doi          = {10.1109/FOCS52979.2021.00050},
  bibsource    = {dblp computer science bibliography, https://dblp.org}
}

@inproceedings{GlinskihR22,
  author       = {Ludmila Glinskih and
                  Artur Riazanov},
  editor       = {Armando Casta{\~{n}}eda and
                  Francisco Rodr{\'{\i}}guez{-}Henr{\'{\i}}quez},
  title        = {{MCSP} is Hard for Read-Once Nondeterministic Branching Programs},
  booktitle    = {{LATIN} 2022: Theoretical Informatics - 15th Latin American Symposium,
                  Guanajuato, Mexico, November 7-11, 2022, Proceedings},
  series       = {Lecture Notes in Computer Science},
  volume       = {13568},
  pages        = {626--640},
  publisher    = {Springer},
  year         = {2022},
  url          = {https://doi.org/10.1007/978-3-031-20624-5\_38},
  doi          = {10.1007/978-3-031-20624-5\_38},
  bibsource    = {dblp computer science bibliography, https://dblp.org}
}

@article{GlinskihR24,
  author       = {Ludmila Glinskih and
                  Artur Riazanov},
  title        = {Partial Minimum Branching Program Size Problem is ETH-hard},
  journal      = {Electron. Colloquium Comput. Complex.},
  volume       = {{TR24-117}},
  year         = {2024},
  url          = {https://eccc.weizmann.ac.il/report/2024/117},
  eprinttype    = {ECCC},
  eprint       = {TR24-117},
  bibsource    = {dblp computer science bibliography, https://dblp.org},
  note         = {(To Appear in ITCS 2025)}
}

@ARTICLE{KleinPaterson80,

  author={Klein and Paterson},

  journal={IEEE Transactions on Computers}, 

  title={Asymptotically Optimal Circuit for a Storage Access Function}, 

  year={1980},

  volume={C-29},

  number={8},

  pages={737-738},

  doi={10.1109/TC.1980.1675657}}

@inproceedings{Paul75,
author = {Paul, Wolfgang J.},
title = {A 2.5 n-lower bound on the combinational complexity of Boolean functions},
year = {1975},
isbn = {9781450374194},
publisher = {Association for Computing Machinery},
address = {New York, NY, USA},
url = {https://doi.org/10.1145/800116.803750},
doi = {10.1145/800116.803750},
booktitle = {Proceedings of the Seventh Annual ACM Symposium on Theory of Computing},
pages = {27–36},
numpages = {10},
location = {Albuquerque, New Mexico, USA},
series = {STOC '75}
}

@article{DBLP:journals/siamcomp/Paul77,
  author       = {Wolfgang J. Paul},
  title        = {A 2.5 n-Lower Bound on the Combinational Complexity of Boolean Functions},
  journal      = {{SIAM} J. Comput.},
  volume       = {6},
  number       = {3},
  pages        = {427--443},
  year         = {1977},
  url          = {https://doi.org/10.1137/0206030},
  doi          = {10.1137/0206030},
  bibsource    = {dblp computer science bibliography, https://dblp.org}
}

@inproceedings{DBLP:conf/coco/MazorP24a,
  author       = {Noam Mazor and
                  Rafael Pass},
  editor       = {Rahul Santhanam},
  title        = {Gap {MCSP} Is Not (Levin) NP-Complete in Obfustopia},
  booktitle    = {39th Computational Complexity Conference, {CCC} 2024, July 22-25,
                  2024, Ann Arbor, MI, {USA}},
  series       = {LIPIcs},
  volume       = {300},
  pages        = {36:1--36:21},
  publisher    = {Schloss Dagstuhl - Leibniz-Zentrum f{\"{u}}r Informatik},
  year         = {2024},
  url          = {https://doi.org/10.4230/LIPIcs.CCC.2024.36},
  doi          = {10.4230/LIPICS.CCC.2024.36},
  bibsource    = {dblp computer science bibliography, https://dblp.org}
}

@inproceedings{DBLP:conf/focs/LiuP20,
  author       = {Yanyi Liu and
                  Rafael Pass},
  editor       = {Sandy Irani},
  title        = {On One-way Functions and Kolmogorov Complexity},
  booktitle    = {61st {IEEE} Annual Symposium on Foundations of Computer Science, {FOCS}
                  2020, Durham, NC, USA, November 16-19, 2020},
  pages        = {1243--1254},
  publisher    = {{IEEE}},
  year         = {2020},
  url          = {https://doi.org/10.1109/FOCS46700.2020.00118},
  doi          = {10.1109/FOCS46700.2020.00118},
  bibsource    = {dblp computer science bibliography, https://dblp.org}
}

@book{CormenLRS,
  author       = {Thomas H. Cormen and
                  Charles E. Leiserson and
                  Ronald L. Rivest and
                  Clifford Stein},
  title        = {Introduction to Algorithms, 3rd Edition},
  publisher    = {{MIT} Press},
  chapter      = {24},
  pages        = {655-657},
  year         = {2009},
  url          = {http://mitpress.mit.edu/books/introduction-algorithms},
  isbn         = {978-0-262-03384-8},
  bibsource    = {dblp computer science bibliography, https://dblp.org}
}
\clearpage

\appendix
\section{Circuits}\label{sec:Circuits}
In this section, we precisely define Boolean circuit complexity and accompanying notions related to the ``shapes'' and ``parts'' of circuits.  These definitions are straightforward, but the details matter because we are working with \emph{general} circuits and functions of \emph{linear} size complexity.  We cannot, for example, afford to layer circuits or put them into negation normal form (NNF, all $\NOT$-gates on the inputs) because size-optimal circuits do not (in general) obey these restrictions.  In particular, optimal circuits for $\XOR$ are neither layered nor in NNF (Theorem \ref{thm:XOR-structure}).  We will require normal forms that are \emph{free} to impose (Lemma \ref{lem:normalization}).

Define general circuits over the basis $\cB = \{\land, \lor, \neg, 0, 1 \}$ of Boolean functions: binary $\land$ and $\lor$, unary $\neg$ and zero-ary (constants) $1$ and $0$. Circuits take zero-ary variables in $X = \{x_1, x_2, \dots , x_n\}$ and $Y=\{y_1, y_2, \ldots y_m\}$ for some fixed $n$ and $m$ as inputs. Throughout this work, we use the standard formulation of circuits consisting of single-sink and multi-source DAGs where internal nodes, called gates, are labeled by function symbols, sources are labeled by input variables, and edges as ``wires'' between the gates. We write $V_C$ and $E_C$ to denote the set of nodes and the set of edges of a circuit $C$ respectively, omitting the subscript when it is clear which circuit is being referenced. Circuits compute Boolean functions through \emph{substitution} followed by \emph{evaluation}. 

An \emph{assignment} $\rho$ of input variables is a mapping from the set of inputs to $\{0,1\}$. To substitute into a circuit according to an assignment $\rho$, each input $x_i$ is replaced by the constant $\rho(x_i)$. For a circuit $G$ and assignment $\rho$ we write $G\hook_\rho$ to denote the circuit obtained by substituting into $G$ according to $\rho$. To evaluate a circuit, the values of interior nodes labeled by function symbols are computed in increasing topological order. Each gate's value is obtained by applying it's function to the value of the node's incoming wires. The output of the circuit overall is the value of its sink. 

Let $\cF_n$ be the family of Boolean functions on $n$ variables. We say a circuit $G$ on $n$ variables computes $g \in \cF_n$ if for all $\alpha \in \{0,1\}^n$, $G(\alpha) = g(\alpha)$. The size of a circuit $G$ is denoted $|G|$ and $CC(g)$, the circuit complexity of a Boolean function $g$, is the minimum size of any circuit computing $g$. We study two notions of circuit size, $\mu_{\cD}$ and $\mu_{\cR}$, which count the number of binary gates and the total number of gates respectively. We decorate any notation with $\cD$ or $\cR$ (e.g. $CC(G)^\cD$) to denote $\mu_D$ and $\mu_R$ whenever necessary, but will omit if the measure is clear from context or if a statement holds for both measures. We say a circuit $G$ computing $g$ is optimal if $|G| = CC(g)$ and, in $\cD$, require that the circuit is \emph{normalized}, i.e. does not contain double-negations and every non-$\neg$ node feeds at most one $\neg$ gate.

We will often order the nodes of a circuit by \emph{depth}, the maximum number of binary gates on any path from the node to the output. We sort a circuits nodes in \emph{decreasing} order by depth, i.e. $\texttt{depth}(u) > \texttt{depth}(v)$ implies $u$ appears before $v$ in our ordering, breaking ties arbitrarily. Such orderings are efficiently realizable as depth is efficiently computable. To ensure we correctly account just binary gates for the depth, we first assign the outgoing edges of each binary gate and negation gate with weight $-1$ and $0$ respectively. Then, we simply take the circuit's underlying DAG, reverse its edges, and compute the shortest path from the output node to every other node in linear time \cite{CormenLRS}.

We will sometimes work closely with specific parts of circuits. To this end we formally define a subcircuit rooted at a vertex $\alpha$.
\begin{definition}[Induced Subcircuit rooted at $\alpha$] 
Let $C = (V, E)$ be a circuit and let $\alpha \in V$. Let $P_{\alpha} = \{b \in V | \exists \text{ a path from } \beta \text{ to } \alpha \}$. The induced subcircuit rooted at $\alpha$ in $C$, denoted $C[\alpha]$, is the vertex-induced subgraph of $C$ whose vertex set is $P_\alpha$ and whose edge set consists of the edges in $E_C$ whose endpoints are nodes in $P_\alpha$.
\end{definition}
\section{Gate Elimination}
\label{sec:neo-gate-elimination}

Most of the proofs in this paper involve arguments by gate elimination: we take a circuit, substitute a subset of inputs with constants, and then eliminate gates according to a set of basic simplification rules. This is illustrated in Figure \ref{fig:gate-elim-example}.

\begin{figure}[H]
    \centering
    \includegraphics[]{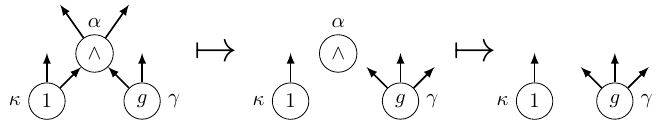}
    \caption{An example of applying the passing simplification $1 \land \gamma \to \gamma$. Notice that $\gamma$ inherits the fanout from $\alpha$, and $\alpha$ can then be garbage collected.}
    \label{fig:gate-elim-example}
\end{figure}

Towards algorithms for the $f$-Simple Extension Problem, we fix a concrete encoding of these circuit-manipulation steps.  Gates $\gamma$ are identified by natural numbers.  We have one book-keeping manipulation: \emph{garabage collection} deletes a gate $\gamma$ with no outgoing wires; i.e, fanout($\gamma$) = 0.  A garbage collection step is encoded by $\langle \textsc{Gc} ,~ \gamma \in \mathbb{N} \rangle$.

The main \emph{gate elimination} rules are organized into four categories in Table \ref{tab:rules}, according to how they impact fanout. All these rules (i) eliminate at least one logic gate and (ii) have no effect on the function computed by a circuit $C$, because they are Boolean identities.

    \begin{table}[t]
    \centering
    \caption{Gate elimination rules, categorized by their impact on fanout}
    \label{tab:rules}
    \newcolumntype{Y}{>{\raggedright\arraybackslash}X}
    \renewcommand\tabularxcolumn[1]{m{#1}}
    \renewcommand{\arraystretch}{1.2}
    \begin{tabularx}{\textwidth}{YYYY}
    Fixing & Passing & Resolving & Pruning \\
    $0 \land \gamma \to 0$ & $1 \land \gamma \to \gamma$ & $\gamma \land \neg \gamma \to 0$ & $\gamma \land \gamma \to \gamma$ \\
    $\gamma \land 0 \to 0$ & $\gamma \land 1 \to \gamma$ & $\neg \gamma \land \gamma \to 0$ & \\
    $1 \lor \gamma \to 1$ & $0 \lor \gamma \to \gamma$ & $\gamma \lor \neg \gamma \to 1$ & $\gamma \lor \gamma \to \gamma$ \\
    $\gamma \lor 1 \to 1$ & $\gamma \lor 0 \to \gamma$ & $\neg \gamma \lor \gamma \to 1$ & \\
    $\neg 0 \to 1$ & & & $\neg \neg \gamma \to \neg$ \\
    $\neg 1 \to 0$ & & &
    \end{tabularx}
  \end{table}
  
  A gate elimination step is encoded by the tuple $\langle \textsc{Ge} ,~ r \in 17 ,~ b : \mathbb{N} \to 7 \rangle$ where $r$ identifies a rule and $b$ is a ``binding'' that maps gates of $C$ to gates of the circuit \emph{pattern} on the left hand side of rule $r$.  Each pattern is a well-formed Boolean formula and thus has a \emph{main connective} $\alpha$ in the standard sense. To apply a rule, examine the right hand side: if it is
  \begin{itemize}
  \item \textbf{a constant $c \in \{0,1\}$:} delete all incoming wires to $\alpha$ and change the type of $\alpha$ to $c$.
  \item \textbf{a matched node (i.e., $\gamma$):} redirect all wires reading $\alpha$ to read from $\gamma$ instead.
  \end{itemize}

  Note that the identifier of $\alpha$ \textbf{never} changes; only its type and incident wires.  Depending on the initial fanout of each gate involved in the pattern, applying a rule could totally disconnect some gate(s) and leave the circuit in a state that needs garbage collection.  We introduce notation for applying sequences of these steps to circuits, and define composed circuit manipulations that are ``well behaved'' with respect to specific circuits.
  
\begin{definition}
  \label{def:simp}
  A \emph{simplification} is a sequence of gate elimination and garbage collection steps.  If $\lambda$ is a simplification and $C$ is a circuit, write $\lambda(C)$ to denote the new circuit obtained by applying each step of $\lambda$ to $C$ in order.  If a single step of $\lambda$ fails to apply in $C$, then fix $\lambda(C) = C$ so that invalid $\lambda$ for $C$ are idempotent by convention.  A simplification $\lambda$ is called
    \begin{itemize}
    \item \emph{terminal} for $C$ if, for every $\lambda'$ that extends $\lambda$ by one additional gate elimination or garbage collection step, $\lambda(C) = \lambda'(C)$, and/or

    \item\emph{layered} for $C$ if the depth of \textbf{binary} gates binding to $\alpha$ in each step is non-decreasing.
    \end{itemize}
  \end{definition}

  Simplifications formalize every sequence of circuit-manipulation steps except for substitution.  They change the function computed by a circuit if and only if some input gate is garbage-collected.  Simplifications suffice to impose convenient structural properties on arbitrary circuits.

  \begin{definition}
    \label{def:normal-ckt}
    A circuit $C$ is \emph{normalized} or \emph{in normal form} if
    \begin{enumerate}
    \item $C$ is constant-free \textbf{or} $C$ is a single constant,
    \item every gate in $C$ has a path to the output, and
    \item no sub-circuit of $C$ matches the left hand side of a gate elimination rule.
    \end{enumerate}
  \end{definition}

  Any terminal simplification suffices to put a circuit $C$ in normal form, and it straightforward to see that the number of constants in $C$ lower-bounds the number of eliminated gates.  We require more: every circuit $C$ can be normalized by a \emph{layered} simplification, and the number of eliminated gates is lower-bounded by the \emph{cumulative fanout} of constants in $C$.

\begin{table}[t]
\centering
\caption{Changes to fanout after each type of gate elimination step.}
\label{tab:fanout-updates}
\renewcommand\tabularxcolumn[1]{m{#1}}
\renewcommand{\arraystretch}{1.2}
\begin{tabularx}{\textwidth}{llll}
Fixing & Passing & Resolving & Pruning \\
$\fo'(\kappa) = \fo(\kappa) - 1$ & $\fo'(\kappa) = \fo(\kappa) - 1$ & $\fo'(\alpha') = \fo(\alpha)$ & $\fo'(\gamma) \geq \fo(\gamma) + \fo(\alpha) - 2$ \\
$\fo'(\gamma)=\fo(\gamma)-1$ & $\fo'(\gamma) = \fo(\gamma) + \fo(\alpha) -1$ & $\fo'(\gamma) = \fo(\gamma) - 1$ & $\fo'(\alpha) = 0$ \\
$\fo'(\alpha) = \fo(\alpha)$ & $\fo'(\alpha) = 0$ & $\fo'(\neg) = \fo(\neg) - 1$ & $\fo'(\neg) = \fo(\neg) - 1$
\end{tabularx}
\end{table}

\begin{lemma}
  \label{lem:normalization}
  For any well-formed circuit $C$ with $q$ constants that have total fanout $\ell$, there is a terminal and layered simplification $\lambda$ such that $\lambda(C)$ is a single constant \textbf{or} $|\lambda(C)| \leq |C| - \ell$.
\end{lemma}

\begin{proof}
  We update two potential functions as $C$ is simplified: $\phi_t$, the cumulative fanout of constants in $C$ and $\mu_t$, the number of binary gates eliminated after $t$ manipulations. Since a circuit is normalized only when it is a single constant or $\phi_t = 0$, it suffices to exhibit a terminal layered simplification $\lambda$ such that $\mu_t \geq \ell$ if $\phi_t = 0$.
  
  Sort the gates of $C$ by depth (breaking ties arbitrarily) to fix a total order on gates.  This ordering remains consistent throughout the entirety of simplification.  Constructing $\lambda$ is straightforward: apply gate elimination rules in depth-order (by $\alpha$) from the input gates of the circuit ``up'', garbage collecting any disconnected gate(s) immediately afterwards so only gates with a path to the output remain.
  
  To update $\phi_t$, we enumerate how fanout changes for the possible gates that may appear in each rule: the main connective $\alpha$, a matched node $\gamma$, a constant $\kappa$, and a negation gate $\neg$.  Write $\fo(\cdot)$ for the fanout of a gate before rule application, and $\fo'(\cdot)$ for the fanout of a gate afterwards.  For example, in Figure \ref{fig:gate-elim-example}, we have $\fo'(\alpha') = 0 ,~ \fo'(\kappa) = \fo(\kappa) - 1$, and $\fo'(\gamma) = \fo(\gamma) + \fo(\alpha) - 1$.  We display the fanout-updates for each rule type in Table \ref{tab:fanout-updates}.  

  Because $\alpha$ has a path to the output, $\fo(\alpha) \geq 1$.  Thus, each rule application reduces $\phi_t$ by at most $1$, even when $\gamma$ matches a constant (as in Figure \ref{fig:fix-with-constant-example}).  Any reduction in $\phi_t$ by $1$ is accompanied by a corresponding increase in $\mu_t$. Furthermore, observe that any subsequent garbage collection steps only increase $\mu$ and never reduce the $\phi$. Any $\beta$ that is garbage collected cannot read a constant: $\beta$ would be the main connective for a rule application lower in the circuit than $\alpha$. Hence if $\phi_t = 0$, we must have $\mu_t \geq \ell$.

Lastly, always choosing the depth-maximal $\alpha$ guarantees that our simplification is layered. This follows from the fact that rule applications only introduce new opportunities for gate elimination whose main connectives are shallower in the circuit.
\end{proof}

\begin{figure}[b]
    \centering
    \includegraphics[]{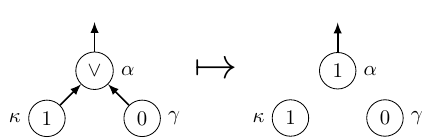}
    \caption{An application of a pruning rule where $\gamma$ itself is a constant. Since $\fo(\alpha) \geq 1$, we have $\phi_{t+1} = \phi_t + \fo'(\alpha) - 2 = \phi_t + \fo(\alpha) - 2 \geq \phi_t + 1 - 2 = \phi_t - 1.$}
    \label{fig:fix-with-constant-example}
\end{figure}

To finish capturing circuit manipulation, we encode a variable substitution step by the tuple $\langle\textsc{Sub} ,~ v \in \{x, y\} ,~ i \in \mathbb{N} ,~ c \in \{0, 1\} \rangle$ where $v$ identifies the target set of variables (base $x$'s or extension $y$'s), $i$ gives the variable number, and $c$ is the value to be substituted.  Apply a single substitution $\langle\textsc{Sub} ,~ v ,~ i ,~ c \rangle$ and to the circuit $C$ by changing the type of every $v_i$-gate $\gamma$ to the given constant $c$ while \textbf{leaving the identifier of each $\gamma$ intact,} exactly as in the application of a fixing rule.

Substitution \emph{always} changes the function computed by $C$, restricting its domain to a subcube.  Additionally, $C$ remains well-formed after any substitution, because input gates have fanin $0$.  However, $C$ is \emph{never} in normal form immediately after a substitution because it introduces a constant.  So we generalize simplifications by allowing for substitution steps.

\begin{definition}
  A \emph{restriction} is a sequence of substitution, gate elimination, and garbage collection steps.  Restrictions generalize simplifications, so we extend notation and conventions accordingly:  $\rho(C)$ denotes the result of applying restriction $\rho$ to circuit $C$, and invalid restrictions for $C$ are idempotent by convention.
\end{definition}

The categories of \emph{terminal} and \emph{layered}
simplifications apply immediately to restrictions,  because they explicitly reference only gate elimination or garbage collection steps.  Thus, restrictions are terminal and/or layered only with respect to how they simplify and \textbf{not} how they substitute.  When applying a restriction $\rho$ to a Boolean function instead of a circuit, we simply ignore the additional simplification steps (if any).
\section{Structural Properties of Optimal Simple Extension Circuits}
\label{sec:SE-tools}

In this section we characterize the structure of optimal simple extension circuits.  The main result is a \emph{$Y$-trees decomposition} (Theorem \ref{thm:y-tree-decomposition}), summarized below.

\begin{theorem*}[Informal]
In any optimal circuit for a simple extension, extension variables occur in isolated read-once formulas \emph{($Y$-trees)} that depend only on extension variables.
\end{theorem*}

\begin{remark} From this point onward, we exclusively consider the DeMorgan basis $\cD$. The same arguments show an identical decomposition theorem in $\cR$. In fact, $\cR$ is simpler. Negations cannot be involved in any of our restrictions with keys, since Lemma \ref{lem:normalization} removes at least $m$ \emph{binary gates} from the circuit and no simplification steps introduce negations. Therefore, the removal of any negations would surpass the budgeted additional $m$ gates in our complexity measure. Consequently, each $Y$-tree in an optimal $\cR$ simple extension circuit is not only a read-once formula, it is also monotone.
\end{remark}

\subsection{Restrictions With ``Speed Limits'' on Gate Elimination}

First we construct restrictions that are guaranteed to eliminate the gates of an optimal simple extension circuit as ``slowly'' as possible: exactly one costly logic gate is eliminated for each substituted variable.  Structural information is then obtained by ``watching'' these slow restrictions operate.  We begin by establishing basic properties of optimal simple extension circuits.  First, combining normalization with the requirement that simple extensions are non-degenerate functions yields

\begin{corollary}[Extension Variables are Read-Once]\label{cor:extension-read-once}
    In any optimal circuit for a simple extension, the fanout of each extension variable is exactly one. Furthermore, if an extension is negated, this negation has fanout exactly one.
\end{corollary}

If any extension variable has fanout greater than one, substituting a key and normalizing (Lemma \ref{lem:normalization}) would eliminate more than $m$ gates. This leads to a contradiction: the resulting circuit computing $f$ would have size less than $CC(f)$.  Generalizing this observation, we obtain a short list of feasibly-checkable properties that guarantee a given circuit is the optimal circuit of a simple extension in Lemma \ref{lem:no-false-witnesses}.

Essentially, the existence of a circuit $G$ computing $g$ with size $CC(f) + m$ that (1) is not ``trivially sub-optimal'' and (2) not ``trivially degenerate'' witnesses that $g$ is a Simple Extension of $f$.  This is not immediate; such a circuit does not necessarily \emph{guarantee} non-degeneracy of $g$ and that $CC(g) = CC(f) + m$ as $G$ could admit non-trivial simplifications.  However, non-degeneracy of $f$ implies that such straightforward witness circuits suffice.

\begin{lemma}
    \label{lem:no-false-witnesses}
    Let $f \in \cF_n$ be non-degenerate and suppose $g \in \cF_{n+m}$ is an \emph{arbitrary} extension of $f$. If there exists a normalized circuit $G$ computing $g$ such that
    \begin{enumerate}
    \item $|G| = CC(f) + m$,
    \item $G$ reads each base variable at least once, and
    \item $G$ reads each extension variable exactly once,
    \end{enumerate}
    then $g$ is a \emph{simple} extension of $f$ and $G$ is an optimal circuit for $g$.
\end{lemma}

\begin{proof}
Let functions $f,g$ and circuit $G$ be as specified in the lemma.  If $g$ is non-degenerate in all variables, then $G$ \emph{must} be an optimal circuit for $g$ --- otherwise we could restrict $G$ with any key to $f$ in $g$ and contradict the circuit complexity of $f$. So in this case $g$ is indeed a simple extension of $f$, and $G$ is optimal. Suppose now (towards contradiction) that $g$ is degenerate with respect to some set of extension variables $D$.

We proceed by reverse induction starting with base case $|D| = 1$. Let $y_i$ be the unique degenerate variable.  By definition of degeneracy, $g|_{y_i = 0} = g|_{y_i = 1}$.  Thus, there are keys to $f$ in $g$ with both $y_i = 1$ and $y_i = 0$.  Therefore, we can substitute $y_i$ to eliminate $\geq 1$ costly logic gate $\alpha$ of $G$ with a \emph{fixing rule.}  The matched gate $\gamma$ loses a wire to $\alpha$, so it may become disconnected.  Run garbage collection recursively starting at $\gamma$, until no more fanout-0 gates exist in $G$, and call the resulting circuit $G'$.

Because $y_i$ is the \emph{only} degenerate variable, no other input gates can be disconnected by garbage collection after substituting into $y_i$ (though some logic gates may be eliminated). 
So $G'$ is a well-formed circuit with constant fanout $\geq 1$ and size $|G'| \leq |G| - 1$.  Furthermore, $G'$ computes function $g'$, an $(n + m - 1)$-variable extension of $f$, because the keys to $f$  in $g$ matching one setting of $y_i$ are preserved in $g'$.  Now substitute the remaining $(m-1)$ extension variables of $G'$ according to one of these keys to $f$ in $g'$ and normalize (Lemma \ref{lem:normalization}) to obtain a circuit $F^*$ that either
\begin{enumerate}
\item has at most $|G'| - (m - 1) + 1 < CC(f)$ gates, or
\item computes a constant function. 
\end{enumerate}
Both cases yield a contradiction; either to the circuit complexity or non-degeneracy of $f$.

Now suppose $|D| \geq 1$.  Because all $y$-variables in $D$ are degenerate, there are keys for all possible settings of them.  So we reduce to the base case by substituting $|D| - 1$ extension variables such that, for each variable, exactly one costly logic gate is eliminated by a \emph{passing rule.}  Passing rules never disconnect the matched $\gamma$.  This leaves a single degenerate $y$-variable in the new function $g''$, an $(n + m - |D| + 1)$-variable extension of $f$, computed by circuit $G''$ obtained by normalization.  $G''$ and $g''$ are now exactly as specified in the statement of this Lemma, except for the guarantee that $g''$ has a unique degenerate extension variable.  Therefore the base case with $|D| = 1$ applies to $G''$ and $g''$, concluding this proof.
\end{proof}

We will want to substitute key bits one by one to better analyze intermediate circuits. This analysis lends itself to induction only if the intermediate circuits themselves are optimal circuits for intermediate simple extensions.  This motivates the following special class of restrictions.

\begin{definition}[All-Stops Restrictions]
  Let $C$ denote an arbitrary circuit with $n+m$ input variables.  A restriction $\rho$ that substitutes $m$ variables is \emph{all-stops} for $C$ if, for each $i \in m$, there is a prefix $\rho_i$ of $\rho$ such that $\rho_i(C)$ is an optimal circuit for a simple extension of the function computed by $\rho(C)$ on $(n + m - i)$ variables.
\end{definition}

If a restriction is all-stops, then exactly one binary gate is eliminated per substitution. Since each extension variable (or its negation) is read by a single binary gate, it is straightforward to analyze the simplifications that occur.

\begin{claim}[Simple Simplifications]\label{cor:all-stop-rewrites-are-simple}
    Simplifications between substitutions in all-stops restrictions are \emph{simple}, i.e. after substitution, the following gate elimination rules are applied in this exact order: a constant negation (if necessary), a passing rule, and lastly a double negation (if necessary).
\end{claim}

The sole binary gate cannot be eliminated via a fixing rule, as otherwise the circuit is not constant free and must either become constant or we could eliminate more costly gates. These highly-structured manipulations will help us build up a robust structure for optimal simple extension circuits in the next subsection. However, we must first show that such restrictions even exist. We surpass this goal: we construct one that is also layered.

\begin{lemma}\label{lem:exists-all-stops}
  Let $f \in \mathcal{F}_n$ and suppose $g \in \mathcal{F}_{n + m}$ is a simple extension of $f$.  For any optimal circuit $G$ computing $g$, there is an all-stops restriction $\rho$ for $G$ such that
  \begin{enumerate}
  \item $\rho$ is terminal and layered, and
  \item $\rho(G)$ is an optimal circuit computing $f$.
  \end{enumerate}
\end{lemma}

\begin{proof}
  Let Boolean functions $f, g$ and circuit $G$ be as in the Lemma.  Sort the binary gates of $G$ by depth (breaking ties arbitrarily) to fix a total order on gates. This ordering remains consistent throughout the entirety of our proof.  Denote by $K \subseteq \{0,1\}^m$ the set of all keys to $f$ in $g$.  Normalize $G$ to eliminate double negations --- no other rules can apply because $G$ is optimal (Lemma \ref{lem:normalization}).  Record these initial steps in $\rho_0$, which we then iteratively extend to the desired all-stops restriction.

  Take $y_i$, read by the maximal in $G$.  Such $y_i$ are well-defined because $G$ reads all extension variables exactly once (Lemma \ref{lem:no-false-witnesses}) and so each $y_i$ has a unique depth in $G$.  Consider the subcircuit around $y_i$: denote by $\alpha$ the \textbf{binary} main connective that reads $(\neg)y_i$ and by $\gamma$ the sibling of $(\neg)y_i$ (the ``matched gate'' in elimination rules).  Now let $K' \subseteq K$ be the subset of keys such that substituting $y_i$ according to any $k \in K'$ would eliminate $\alpha$ with a passing rule.

  If $K'$ is non-empty, then substitute $y_i$ according to any key in $K'$ and normalize the resulting circuit.  Record these steps in a restriction $\rho_1$ that extends $\rho_0$.  The resulting circuit $G' = \rho_1(G)$ is in normal form, reads each base variable exactly as many times as $G$ did, and reads each remaining extension variable exactly once.  Therefore, the function computed by $G'$ is a simple extension of $f$ on $(n + m - 1)$ variables, and $G'$ is an optimal circuit (Lemma \ref{lem:no-false-witnesses}).  By construction, $\rho_1$ is layered and terminal, so we will continue extending $\rho_1$  by selecting a depth-maximal extension variable in $G'$ and re-starting this case analysis.
  
  If $K'$ is empty, consider cases depending on the type of $\gamma$.  Suppose towards contradiction that $\gamma$ is not an extension variable or the negation of an extension variable; in this case there is no $y_j$ with a path to $\gamma$ because we selected a $y_i$ of maximal depth.  Substitute $y_i$ according to any key in $K$, and observe that we obtain a circuit with total constant-fanout exactly $\fo(\alpha) \geq 1$ after applying a fixing rule.  Then substitute all other extension variables according to the same key: the resulting circuit $G''$ has total constant-fanout $\fo(\alpha) + m - 1$.  Normalizing, we contradict either the circuit size or non-degeneracy of $f$ (exactly as in the proof of Lemma 16).

  Therefore, if $K'$ is empty, we must have that $\gamma$ is $(\neg)y_j$ for some $j \neq i$ --- a \emph{layer tie} between extension variables. Thus, a restriction will be layered regardless of which substitution (into $y_i$ or $y_j$) is made first. Using the particular structure of this sub-circuit (i.e., ``$~ (\neg)y_j ~ \alpha ~ (\neg)y_i ~$'' where $\alpha$ is the main binary connective) we construct a key that sets $y_j$ to eliminate $\alpha$ with a passing rule while preserving the rest of $K$.  The idea is to extract information from $G$ after $y_i$ is substituted according to any key in $K$. Such substitution mutates the sub-circuit around $y_i$ in $G$.  If we can \emph{identically mutate} $G$ after setting $y_j$ to eliminate $\alpha$ by passing \emph{first}, then we are done.  The details of this construction follow.

  Purely to inspect the results, substitute $y_i$ in $G$ according to any key in $K$, and observe that we obtain a circuit $G'$ with total constant fanout exactly $\fo(\alpha)$ after applying a fixing rule.  Furthermore, the function computed by $G'$ is an $(m-1)$-input extension of $f$ where $\alpha$ has been substituted with a particular constant value $c_i \in \{0,1\}$ and the variable $y_j$ is disconnected from the input (and thus degenerate).  This constant $c_i$ is all the information we need.

  Returning to $G$, there \emph{must} exist a substitution $s_j \in \{0,1\}$ for $y_j$ that eliminates $\alpha$ via a passing rule: this follows immediately by inspecting the truth tables of binary $\{\land, \lor\}$ and case analysis on the occurrence of a negation gate between $y_j$ and $\alpha$. Extend $\rho$ by first substituting $s_j$ for $y_j$ in $G$ and normalize.  The resulting circuit $G''$ passes all wires reading $\alpha$ to  $(\neg)y_i$ and loses exactly one costly gate.  Again by inspection, we can substitute $y_i$ with some value $s_i \in \{0,1\}$ such that the wires reading $(\neg)y_i$ read the constant $c_i$ --- identical to the subcircuit in $G'$.  Thus, $G'$ computes an $(m-1)$-input extension $g_j$ of $f$ witnessed by the key $\{ y_j \mapsto s_j \}$.  Extend the working restriction $\rho_1$ by first $\{ y_j \mapsto s_j \}$ and then $\{ y_i \mapsto s_i \}$, recording normalization steps in between and afterwards.

Finally, observe that $G'$ meets all the preconditions of Lemma \ref{lem:no-false-witnesses} and so $g_j$ is a simple extension of $f$ and $G'$ is an optimal circuit computing $g_j$.  The proof concludes by iterating this case analysis on the set of keys and depth-maximal extension variables until all extension variables are eliminated.
\end{proof}

\subsection{\texorpdfstring{$Y$}{Y}-tree Decomposition}
Equipped with layered all-stops restrictions, we can now inductively prove that optimal simple extension circuits have the following nice structure: each extension variable occurs in an isolated read-once subformula called a $Y$-tree that depends only on other extension variables, and the output of these read-once formulas is read by a single costly gate called a combiner. Formally,

\begin{definition}[Y-Tree Decomposition]\label{def:y-tree-decomposition}
  Let $G$ be a circuit with two distinguished sets of inputs: \emph{base variables} $X$ and \emph{extension variables} $Y$. We say a triple $\langle \delta , b , T \rangle$, where $\delta$---referred to as a \emph{combiner}---is a binary gate in $G$, bit $b \in \bool$ designates an input of $\delta$, and $T$ is a sub-circuit of $G$ rooted at the $b$ child of $\delta$, is \emph{admissible} if the following local conditions are met:
    \begin{enumerate}
      \item Each $T$ is a read-once formula in only extension variables $Y$.
      \item Each $T$ is \emph{isolated} in $G$ --- gate $\gamma$ is the unique gate reading from $T$, and it only reads the root of $T$.
      \item The sub-circuit of $G$ rooted at the $\neg$b child of $\gamma$ contains at least one $X$ variable.
    \end{enumerate}
  
  A \emph{$Y$-Tree Decomposition} of $G$ is a set of admissible triples which are \emph{non-intersecting}, that is, each $y_i \in Y$ appears in at most one $T$. The \emph{size} of a decomposition is the number of tuples present in the decomposition. The \emph{weight} of a $Y$-tree decomposition is the number of extension variables that are read in some $T$. We say a $Y$-tree decomposition is \emph{total} if its weight is $|Y|$, i.e. every extension variable appears.
\end{definition}

We first remark that the binary gates spread across the read-once formulas and combiners in a total decomposition account for all $m$ binary gates eliminated when substituting and simplifying with a key.

\begin{theorem}\label{thm:y-tree-decomposition}
  If $G$ is a minimal circuit for a simple extension, then $G$ has a total $Y$-tree decomposition.
\end{theorem}

\begin{proof}
    We proceed via induction over $m$, the number of extension variables. When $m = 0$, the statement is vacuously true: the empty $Y$-tree decomposition is total for any optimal circuit of the base function. 
    
    Assume the statement holds for any simple extension with $k-1$ extra variables. Let $G_{k}$ be an optimal circuit for $g_{k}$, some simple extension of a function $f$, with $k$ extension variables. By Lemma \ref{lem:exists-all-stops}, there exists an all-stops layered restriction $\rho$ such that $\rho(G_{k})$ is optimal for $f$. By definition, there exists a prefix $\rho_1$ of $\rho$ such that $G_{k-1} = \rho_1(G_k)$ is an optimal circuit for a simple extension of $f$ with $k-1$ extension variables. Without loss of generality, assume $y_k$ is the variable substituted and eliminated by $\rho_1$. By our inductive hypothesis, $G_{k-1}$ admits a total $Y$-tree decomposition, $\mathcal{Y}_{k-1}$. We will show how to add to/modify $\mathcal{Y}_{k-1}$ to obtain a total $Y$-tree decomposition for $G_k$ by carefully considering the steps of $\rho_1$, .
    
    By Corollary \ref{cor:all-stop-rewrites-are-simple}, $\rho_1$ is exactly the following sequence: a single substitution, (possibly) a constant negation, a passing rule, and (possibly) a double negation elimination. Consider this passing rule, and denote by $\alpha$ its the main connective, $\kappa$ its constant, and $\gamma$ its matched node. Observe that $(\neg)y_k$ is $\kappa$ since, after substitution (and possibly constant negation), it is the only constant in the circuit. Let $(\neg)\beta$ be the matched node $\gamma$, i.e. the other input to $\alpha$. We remark that $\rho_1$ only impacts the local neighborhoods of nodes adjacent to $\alpha$; any triples of $\mathcal{Y}_{k-1}$ not involving these nodes are still admissible in $G_k$. We modify $\mathcal{Y}_{k-1}$ in the following ways, depending on whether $\beta$ is present in any triples.
    \begin{description}
    \item[Adding a Triple:] If $\beta$ is not present in any triple, we argue that the triple $\langle \alpha, b, T'=(\neg)y_i\rangle$, where $b$ denotes which input of $\alpha$ is $(\neg)y_i$, can be added to $\mathcal{Y}_{k-1}$. It is easy to see this triple is admissible: since $\rho$ is layered, $\beta$ can only depend on $X$ variables, else there are deeper nodes reading extension variables that must be eliminated first. However, we must check that all triples of $\mathcal{Y}_{k-1}$ remain admissible in $G_k$. The only nodes of $\mathcal{Y}_{k-1}$ possibly impacted from $G_k$ to $G_{k-1}$ are combiners, since only those can read $(\neg)\beta$ in $G_{k-1}$. However, these combiners still fulfill the third condition in $G_k$: the sub-circuit rooted at $(\neg)\alpha$ still depends contains $X$ variables since $(\neg)\alpha$ reads $(\neg)\beta$. Hence all of these triples are admissible, and are non-intersecting---forming a total $Y$ decomposition for $G_k$.

    \item[Modifying a Triple:] It is not hard to see that the non-intersecting requirement enforces that $\beta$ appear in at most one triple. Furthermore, if $\beta$ is present in some triple, then it must \emph{be} an extension variable: if it is an interior binary gate of some tree $T$, then there are deeper binary gates than $\alpha$ that read extension variables which are eliminated later in the layered $\rho$. When $\beta$ is present in some triple of $\mathcal{Y}_{k-1}$, we will need to modify its tree to include $y_k$ and $\alpha$ 
    
    Let $t = \langle\delta_i,b_i, T_i\rangle$ be the triple from $\mathcal{Y}_{k-1}$ that includes $\beta$. Not only is $t$ inadmissible in $G_k$, it is not even well-formed: wires reading nodes of $T_i$ in $G_{k-1}$ are actually reading $(\neg)\alpha$ in $G_k$. As the modifications to $G_k$ by $\rho_1$ are local, every other triple in $\mathcal{Y}_{k-1}$ remains admissible in $G_k$. We simply need to incorporate $(\neg)\alpha$ and $(\neg)\gamma$ into the triple, ensuring the first two conditions are met. Observe that there is only one node reading $(\neg)\alpha$, as otherwise the fanout of $(\neg)\beta$ will have greater than one fanout in $G_{k-1}$, violating Corollary \ref{cor:extension-read-once}. We modify $t$ in two ways, depending on which gate in the triple is reading it.
    \begin{description}
        \item[$\alpha$ disconnects $\delta_i$ from $T_i$:] If $\delta_j$ reads $(\neg)\alpha$, then observe it, not the root of $T_i$, is the $b_i$ child of $\delta_j$. However, the sub-circuit $T'$ rooted at $(\neg)(\alpha)$ is a read-once formula in only extension variables. Replacing $T_i$ with $T'$ in $t$ repairs the triple and produces a total decomposition. 
        
        \item[$\alpha$ disconnects $\beta$ inside $T_i$:] If a node in $T_i$ reads $(\neg)\alpha$ in $G_k$, then $T_i$ is not actually a well-formed sub-circuit of $G_k$. However, the sub-circuit $T'$ rooted at the $b_i$ child of $\delta_i$ is still read-once formula over extension variables, once we include $(\neg)\alpha$ and $(\neg)y_k$. Hence $\langle\gamma_j, b_j, T'\rangle$ can replace $t$ to produce a total decomposition $\mathcal{Y}_k$. 
    \end{description}
    \end{description}
\end{proof}
\section{The \texorpdfstring{$f$}{f}-Simple Extension Problem is Fixed-Parameter Tractable}\label{sec:xor-se-solver}

In this section we build a polynomial time algorithm for the $\XOR$-Simple Extension Problem. By Lemma \ref{lem:no-false-witnesses}, it suffices to find a normalized circuit for $g$ of size $CC(f) + m$ which reads every extension variable. 

For exposition, recall the framing of encoding simple circuit extensions from Section \ref{sec:fpt-sep-solver}: suppose $g$ is a simple extension of $f$ and Alice knows $G$, an optimal circuit for $g$.  Alice can obtain an optimal circuit $F$ computing $f$ by simply restricting the $y$-variables of $G$ with a key and performing gate elimination.  Now consider the following communication problem: Bob knows $F$, and Alice would like to send him $G$ using as few bits as possible.  Because $g$ is a simple extension of $f$,  Alice can compute the $Y$-tree decomposition of optimal circuit $G$.  The idea is to send Bob a sequence of instructions that tell him exactly how to graft each Y-Tree of $G$ onto the gates of $F$.

\subsection{Grafting \texorpdfstring{$Y$}{Y}-Trees: Key Ideas \& Structural Properties}

We begin by illustrating the ``grafting'' idea for the algorithm.  First consider the case where there is only a single Y-tree $T_\mathcal{Y}$ in the decomposition of $G$.  Then there must be some costly gate $\eta$ --- an \emph{origin} already present in $F$ --- that is combined with $T_\mathcal{Y}$ in $G$.  Our decomposition theorem shows that every possible arrangement of this \emph{graft} is depicted in Figure \ref{fig:combiner}, which shows the local neighborhood\footnote{Think of this as ``zooming in'' to inspect one of the shaded circles in Figure \ref{fig:example-y-tree}.} of the single Y-tree in $G$.  

\begin{figure}[t]
  \begin{subfigure}{0.45\textwidth}
    \centering
    \includegraphics[]{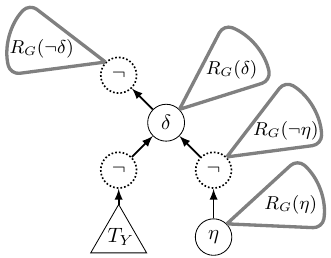}
    \caption{Before any extension variable restrictions}
    \label{subfig:pre-elim}
  \end{subfigure}
  \begin{subfigure}{0.45\textwidth}
    \centering
    \includegraphics[]{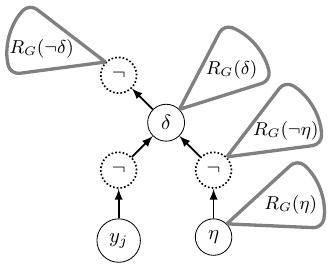}
    \caption{All extension variables restricted except $y_j$}
    \label{subfig:post-elim}
  \end{subfigure}

\caption{The local neighborhood of a combiner $\delta$ and $Y$-tree $T_{\mathcal{Y}}$ grafted on an origin $\eta$.}
\label{fig:combiner}
\end{figure}

In Figure \ref{fig:combiner}, both $\delta$ and $\eta$ represent costly gates that \emph{must} be present, the dotted negations represent NOT gates that \emph{may} be present, and the cones represent connections to the rest of $G$.  The notation $R_G(\beta)$ means ``the set of costly gates that read gate $\beta$ in circuit $G$.''  These sets are depicted by shaded cones, because at least one of them must be non-empty --- otherwise, $G$ would not be an optimal circuit --- but we do not know which.  Note how the $Y$-tree and potential negation atop it are \emph{isolated} from the rest of the circuit, except for connections via the \emph{combiner} gate $\delta$.  We will describe in some detail how this single graft can be transmitted and sketch the issues involved in efficiently coding \emph{all} grafts and Y-trees for Bob.

First, Alice applies gate elimination to $G$ using a prefix of a terminal and layered all-stops restriction (Lemma \ref{lem:exists-all-stops}), eliminating all but one of the $y$-variables in $T_\mathcal{Y}$.  Because $T_\mathcal{Y}$ is an isolated formula in $G$, a single variable $y_j$ is left in $G$ at the end of this process: the local neighborhood transforms from Figure \ref{subfig:pre-elim} to Figure \ref{subfig:post-elim}.  Crucially, all gates outside $T_\mathcal{Y}$ are unaffected.  Now, Alice runs a final step of gate elimination, substituting $y_j$ according to the key.  The result will be as depicted in Figure \ref{fig:origin-F-nhood}: the local neighborhood of $\eta$ in $F$.  The key observation is that \emph{Bob has perfect knowledge of this structure} --- it is simply a sub-circuit of $F$.  If Alice can send (1) an identifier for $\eta$ (2) a ``diff'' between the neighborhood of $\eta$ in $F$ compared to $G$ (Fig. \ref{subfig:post-elim} vs. \ref{fig:origin-F-nhood})  and (3) a description of $T_\cY$ this would suffice for Bob to efficiently transform $F$ into $G$.

\begin{figure}[t]
    \centering
    \includegraphics[]{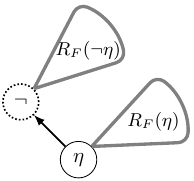}
    \caption{The origin $\eta$ after the $Y$-tree has been pruned.}
    \label{fig:origin-F-nhood}
\end{figure}

The most straightforward protocol would send $O(\log(\#\mathsf{gates}(F)))$ bits to identify $\eta$ for Bob.  This is too expensive if there is more than one Y-tree in $G$.  Recall that we can only tolerate runtime of $2^{O(n + m)}$ and intend to brute-force all possible codes.  There could be as many as $m$ Y-trees with a single variable each and thus $m$ origins, so brute-forcing this simple code would require time $2^{O(m\log(n) + n)}$ for $\#\mathsf{gates}(F) = O(n)$ --- super-polynomial in $2^{(n+m)}$ and therefore unacceptable.

Instead, Alice can send a $\#\mathsf{gates}(F)$-bit origin indicator vector $\chi$, where $\chi_i$ is 1 if gate $i$ of $F$ is the origin of some Y-tree.  It is feasible to brute-force these vectors in $2^{O(n)}$-time for any possible number of origins given a linear upper bound on the circuit size of $f$.  Returning to the case where $G$ has a single Y-tree, Bob identifies $\eta$ by reading the single 1-bit of $\chi$.  For the local neighborhood of $\eta$, the following information must be transmitted to summarize the ``diff'' between $F$ and $G$:
\begin{enumerate}
\item The costly functions computed by gates $\eta$ and $\delta$,
\item presence or absence of each possible NOT gate depicted in Figure \ref{subfig:post-elim},
\item whether $\eta$ is connected to the left or right input of $\delta$,
\item the description of $T_\cY$, and
\item the sets of gates that read from each individual element of the graft, $R_G(\cdot)$.
\end{enumerate}
Items 1 - 3 can be described by a constant-length bitstring, depending only on the enumeration of all possible ``graft'' sub-circuits implied by our characterization of gate elimination.  We code the exact Y-Tree $T_\cY$ (item 4) explicitly for now and eliminate it later.  Here, we are concerned with how to efficiently code the sets $R_G(\cdot)$ without sending explicit ``pointers'' to nodes of $F$, which remain too expensive per the discussion of transmitting $\eta$ above.

To code these wire movements efficiently, Alice will exploit the relationship between the gates reading from $\eta$ in $F$ and the gates reading from $\eta$ in $G$ as determined by gate elimination.  Bob knows the contents of $R_F(\eta)$ and $R_F(\neg \eta)$ --- the sets of costly gates reading from $(\neg)\eta$ in $F$.  Collect the gates that read from the graft, in both $G$ and $F$:
\begin{align*}
  R_G &= R_G(\eta) \cup R_G(\neg \eta) \cup R_G(\delta) \cup \R_G(\neg \delta)  \\
  R_F &= R_F(\eta) \cup R_F(\neg \eta)
\end{align*}
Observe that $R_G \subseteq R_F$ because during the single run of gate elimination that transforms Figure \ref{subfig:post-elim} into Figure \ref{fig:origin-F-nhood}, all the wires reading the eliminated gate $\delta$ are ``inherited'' by $\eta$.  Furthermore, \emph{Bob knows $\eta$ from $\chi$ and thus can reconstruct $R_F$}, so Alice can identify any wire relevant to the graft by coding ``the $j$-th element of $R_F$.''  Here we must apply the assumption that optimal circuits for $f$ have at most constant fanout, independent of $n$, to bound the length of these \emph{relative} wire-identifiers by a constant.  Thus, Alice can identify \emph{exactly which} wires should be moved from $\eta$ to read $(\neg)\delta$ in $G$.  Inspecting Figure \ref{subfig:post-elim} again, she can also code \emph{where} they move using constantly many bits, because there are only two options: reading from $\neg \delta$ or $\delta$ directly.

This discussion has outlined the base case of our encoding/decoding argument, where only a single Y-tree is transmitted to Bob.  Notice that, if there are multiple \emph{combiner-disjoint} Y-trees in $G$, an essentially identical strategy could encode all these Y-trees.  However, reverse gate elimination can create somewhat more complex structures: specifically, we can have combiners $\delta_i$ grafted onto \emph{a distinct combiner} $\delta_j$, instead of a gate $\eta$ that was present in the original circuit $F$.  These \emph{compounded combiners} correspond exactly to the ``adding a triple'' case in the proof that $Y$-Tree decompositions exist (Theorem \ref{thm:y-tree-decomposition}) where $\alpha$ occurs \emph{between} $\beta$ and an existing combiner.

Even so, Alice can use the fact that every combiner $\delta_i$ has a \emph{unique} origin $\eta$ to identify the ``first'' combiner grafted onto $\eta$ and instruct Bob to add subsequent $\eta$-derived combiners in depth-respecting order.  We now prove the required properties of origins and combiners.

\begin{definition}[Original \& Originating Gates]
Let $G$ and $F$ be circuits, $\rho$ be an all-stops restriction, and suppose $\rho(G) = F$.  We call a gate $\beta$ of $G$ \emph{original} if it is not garbage collected by $\rho$ --- i.e., if $\beta \in F$.  Furthermore, let $G$ have a $Y$-tree decomposition with elements $\langle \delta ,~ b ,~ T\rangle$.  The \emph{origin} of each $\delta$ is the depth-maximal original gate $\eta$ such that $X$ variables occur in $G[\eta]$ and there is a path from $\eta$ to the output of $G$ through $\delta$.
\end{definition}

The origin of a combiner can be depth-trivial: a single input gate $x_i$ or the unique output gate are both valid origins.  Intuitively, if $\eta$ is the origin of $\delta$, then $\delta$ ``takes'' wires from $\eta$ when it is grafted into the circuit.  When the maximum fanout of an optimal circuit is constant, this imposes a hard limit on the number of compounded combiners that share the same origin.

\begin{lemma}[Every Combiner has a Unique Origin]
\label{lem:unique-origins}
  Suppose $G$ is a circuit with $Y$-tree decomposition $\cD$, and let $\langle \delta ,~ b ,~ T\rangle$ be an arbitrary element of $\cD$.  Then $\delta$ has exactly one origin in $G$.
\end{lemma}

\begin{proof}
  We argue by induction on the structure of $G$, exploiting that the DeMorgan basis has fan-in two.  We will walk ``down'' the circuit from $\delta$ until an original gate is encountered.  By definition, $\delta$ has only two inputs: input $b$ leads to a $Y$-tree, and input $\neg b$ leads to a sub-circuit containing $X$ variables.  Thus we must go down the $\neg b$ argument to find the origin. Argument $\neg b =$ gate $\beta$ is either another combiner in the decomposition or not --- if $\beta$ is not a combiner, then it is an original gate, and thus unique because the path to $\beta$ from $\delta$ was fully determined.  $\beta$ is topologically maximal because it was the first non-combiner encountered walking ``down'' the circuit in topological order.  Then, if argument $\neg b$ is another combiner, we are done by induction: there is simply another deterministic choice to find a sub-circuit where $X$ variables occur, because argument $b$ of $\beta$ is also a $Y$-tree.
\end{proof}

Finally, observe that there are only constantly-many ways to eliminate a binary gate by a simple simplification (Claim \ref{cor:all-stop-rewrites-are-simple}) as depicted in Figure \ref{fig:widget-diagram}.  This follows by brute-forcing all possible sequences of simple simplifications around eliminating one costly gate, formalizing the idea of a particular ``graft'' sub-circuit from earlier in this section: call each possible realization of sub-circuits a \emph{widget}.

\begin{figure}[h]
    \centering
    \includegraphics[width=\textwidth]{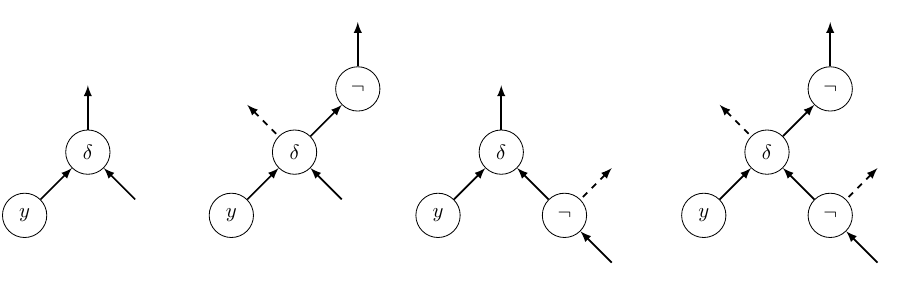}
    \caption{The four possible widgets in a splice code (up to symmetry). The node labeled $y$ represents which child of the combiner $\delta$ is the $Y$-tree. A dotted edge not leading to another node represents one (or more) possible connections to a \emph{costly} gate outside of the widget. A solid edge not leading to another node represents one (or more) guaranteed connections to a \emph{costly} gate outside of the widget.}
    \label{fig:widget-diagram}
\end{figure}

\subsection{Grafting \texorpdfstring{$Y$}{Y}-Trees: Efficient and Explicit Encodings}

Suppose $\eta$ is a gate in some circuit $C$, and let the gates $\zeta_1 , \dots , \zeta_{\leq \ell}$ read $\eta$ in $C$.  We use $\eta$-\emph{relative fanout-coding} to name another gate $\beta$ in a circuit $C'$ obtained by transforming $C$ by writing an $\ell$-bit vector $v$ where:
\begin{equation*}
  v_i(\beta,\eta) =
  \begin{cases}
    1, & \text{if $\zeta_i$ reads $\beta$ in $C'$} \\
    0, & \text{otherwise}
  \end{cases}
\end{equation*}
Trivially, when $C = C'$ and we fanout-code $\eta$ relative to itself, the result is $1^{\fo(\eta)} \circ 0^{\ell - \fo(\eta)}$.  Suitability for fanout-coding is the prime motivation of our requirement that gate identifiers are \textbf{never} mutated (only recycled or discarded) by circuit manipulation.

\begin{definition}[Splice Codes]
  Suppose $F$ is a circuit with maximum fanout $\ell$.  A splice code $E$ is a sequence of tuples that encode a transformation of $F$ into $G$, a circuit for a simple extension of $F$.  The splice code begins with an \textbf{Origin Indicator Vector:} a bitvector of length $CC(f)$ where entry $i$ is set to 1 iff gate $i$ of $F$ is an origin.  Then, for \textbf{each} origin $\eta$, a \textbf{sequence of \emph{splices}} follows --- one splice per combiner originating in $\eta$.  A splice is:
  \begin{enumerate}
  \item \textbf{Target Gate:} Which gate should be grafted onto? Written as $\eta$-fanout-relative code.
  \item \textbf{Selected Wires:} Which wires does the spliced widget take from the target gate?  Written as $\ell$-bit indicator vector, where 1-bits must be a subset of the target gate's code.
  \item \textbf{Widget:} A constant number of bits naming one of constantly-many combiner widgets.
  \item \textbf{Wire Moves:} Map from taken wires taken to gates of the widget: at most $(\ell - 1)$ constant-length identifiers.
  \item \textbf{Explicit $Y$-Tree:} A fully specified read-once formula in the $Y$ variables.
  \end{enumerate}
\end{definition}

\begin{theorem}[Splice Codes for Simple Extensions]
\label{thm:splice-codes}
Suppose $g \in \cF_{n+m}$ is a simple extension of $f \in \cF_n$ and let $\tilde{G}$ be an isomorphism class of minimal circuits for $g$.  Then, there is an ``unlocked'' isomorphism class $\mathsf{UL}(\tilde{G})$ of minimal circuits computing $f$ such that, for every representative $F$ of $\mathsf{UL}(\tilde{G})$, there is a \emph{splice code} $E$ such that an efficient algorithm $\mathtt{Decode}(F, E)$ prints a circuit in the isomorphism class $\tilde{G}$.  The length of $E$ is $O(\ell \cdot |F| \cdot m \cdot \operatorname{bits}(Y))$, where $\ell$ is the maximum fanout of $F$ and $\operatorname{bits}(Y)$ is the number of bits required to code a $Y$-tree on $m$ variables.
\end{theorem}

\begin{proof}
  Take a representative $G$ of the isomorphism class of ciurcits $\tilde{G}$ with gates named such that they are compatible with depth and a topological order.  Let $\mathcal{D}$ be the $Y$-tree decomposition of $G$ (Theorem \ref{thm:y-tree-decomposition}) and let $\rho$ be the particular all-stops restriction of $G$ implicit in the proof of Theorem \ref{thm:y-tree-decomposition} such that $F = \rho(G)$ computes $f$ and $F$ is optimal for $F$.  We will ``reverse'' $\rho$ to obtain a circuit isomorphic to $G$ from $F$.

  First, list all the origins in $G$ --- it is immmediate from the definition and inspecting $\rho$ that this is possible.  Now, if the set of combiners associated with each origin is of size 1, it is clear that the trivial fanout-relative code is used to set the Target Gate for each combiner, and the legnth lower bounds are immediate because the overhead of encoding sequences is linear, and all fields of the sequences are linear in $\ell$, a constant, or proportional to the size of a particular $Y$-tree.  Origin gates are \textbf{never} named explicitly; the order of sequences of splices suffices to associate them with the appropriate splice-sequence.

  Suppose now that some gates of $G$ originate $q$ combiners where $1 < q \leq m$ and fix an arbitrary such origin $\eta$.  To order the splices, perform breadth-first-search starting from $\eta$ in $G$ and observe which wires of $\eta$ have been ``spliced'' with which combiners in each stop of the all-stops restriction $\rho$.  Using this ordering and the sub-sequences of $\rho$ where combiners originating in $\eta$ are eliminated, observe that at each step there is there is a set of maximally-shallow \emph{available combiners} that could be grafted to: those between $\eta$ and the output of the circuit which feed into a gate that used to be fed by $\eta$ directly.

  Each available combiner must be uniquely identified by an $\eta$-relative fanout-code, because otherwise an intermediate normal circuit $\rho$ would not be size-optimal: we could restrict using the key embedded in $\rho$ to find that the same gate must read $\eta$ twice, which triggers a resolving gate-elimination.  Therefore, even compound-combiners can be uniquely regenerated by splice codes.

To conclude the proof, observe that the isomorphism class of $F$ is exactly $\mathsf{UL}(\tilde{G})$, as desired, because all coding operations relied only on being given a circuit isomorphic to $F$ whose gates were named in depth-sorted and thus topological order.  This is because our formalization of circuit manipulation never introduces ``fresh'' identifiers; it only deletes or recycles them.
\end{proof}

\subsection{The Final Ingredient: Truth-Table Isomorphism}

There are two issues with trivial brute-force over splice codes. First, there are too many base circuits for $\XOR_n$: there are at least as many optimal $\XOR_n$ circuits as there are labeled rooted binary trees with $n$ leaves. Let $C_{n}$ be the $n^{\text{th}}$ Catalan number; there are $n!\,C_{n-1} = 2^{\omega(n \log n)}$ labeled rooted binary trees  \cite{van_Lint_Wilson_1992}. Similarly, there are $m!\,C_{m-1}$ explicit $Y$-trees reading all $m$ variables. However, in both expressions, the dominating factorial terms $n!$ and $m!$ arise from permuting the labels of the variables. But two circuits which can be transformed into one another by permuting inputs compute \emph{truth-table isomorphic functions}.

\begin{definition}[\cite{Luks99,ArvindV14}]
    Two functions $f,g \in \cF_{n}$ are \emph{truth-table isomorphic} if there exists a permutation $\pi$ on $[n]$ such that $g(x) = f(\pi(x))$.
\end{definition}
It is straightforward to connect truth-table isomorphism with permuting circuit inputs.
\begin{observation}\label{obs:permuting-variables-truth-table-isomorphism}
The following two statements, the first a universal statement and the second existential, describe the relationship between permuting circuit variable labels and truth-table isomorphism.
\renewcommand\labelenumi{(\theenumi)}
    \begin{enumerate}
        \item  If a Boolean function $f \in \cF_n$ is truth-table isomorphic to another Boolean function $g \in \cF_n$ then \textbf{any circuit} for $f$ can be transformed into some circuit for $g$ by relabeling input $x_i$ with $x_{\sigma(i)}$ for all $i \in [n]$ using some permutation $\sigma$ on $[n]$. 

        \item If \textbf{a given circuit} for $f \in \cF_n$ can be transformed into \emph{a} circuit for $g \in \cF_n$ by relabeling input $x_i$ with $x_{\sigma(i)}$ for all $i \in [n]$ using some permutation $\sigma$ on $[n]$ then $f$ is truth-table isomorphic to $g$.
    \end{enumerate}
\end{observation}

\begin{proof}
    We first prove $(1)$. Let $\pi$ be the witnessing permutation between $f$ and $g$, i.e. $f(x) = g(\pi(x))$. Observe that $f(\pi^{-1}(x)) = g(\pi(\pi^{-1}(x)))=g(x).$ Let $F$ be any circuit for $f$. Let $G$ be the circuit obtained by relabeling each input $x_i$ by $x_{\pi(i)}$. We wish to argue $G$ computes $g$. Observe that evaluating $g$ on input $x$ is the same as evaluating $F$ on $\pi^{-1}(x)$ since $x_{\pi(\pi^{-1}(i))} = x_i$. Therefore $G(x) = F(\pi^{-1}(x)) = f(\pi^{-1}(x)) = g(x)$.
    
    For $(2)$, fix some optimal circuit $F$ and let $\sigma$ be the witnessing permutation used to relabel $F$ to obtain $G$, a circuit for $g$. We have $f(x) = F(x) = G(\sigma^{-1}(x)) = g(\sigma^{-1}(x))$. As $\sigma$ is a permutation on $[n]$, $\sigma^{-1}$ is also a permutation on $[n]$ and thus $f$ and $g$ are truth-table isomorphic.
\end{proof}

The implication is that do not need to try every possible base $\XOR_n$ circuit or even try every explicit splice code. We can take an unlabeled base circuit and unlabeled $Y$-trees and assign variables arbitrarily. We then just need to check whether the function our reconstructed circuit computes is isomorphic to $g$. This is feasible because (rather surprisingly) truth-table isomorphism testing can be done in polynomial time.

\begin{theorem}[Corollary 1.3 of \cite{Luks99}]
\label{thm:tt-iso}
    Given the truth-tables for two Boolean functions, testing whether they are equivalent under permutation of variables can be done in time $c^{O(n)}$ where $c$ is a constant.
\end{theorem}
To this end we formally define ``unlabeled'' optimal circuits.
\begin{definition}[Open Optimal Circuit]
    A circuit is \emph{open} if all of its inputs are unlabeled. An open circuit is optimal for a Boolean function $f$ if there exists some labeling of its' inputs so that the resulting circuit is an optimal circuit for $f$.
\end{definition}

\subsection{An Efficient Simple Extension Problem Solver for \texorpdfstring{$\XOR$}{XOR}}
\label{sec:se-algorithm}
The following algorithm, when given $L$, a complete list of representatives from each open isomorphism class of optimal circuits for $f$, decides the $f$-Simple Extension Problem. By ``implicit splice code,'' we mean a splice code without the $Y$-trees specified.

 \begin{algorithm}[h]
   \caption{$\texttt{Ckt-SE-Solver}(n \in \mathbb{N},~ g \in \cF_{n+m}, L)$}
     \label{alg:ckt-se-solver}
     \begin{algorithmic}[1]
         \State Verify $\exists \rho \in \{0,1\}^m$ such that $g|_\rho \equiv f$ and \textbf{return False} if not
         \State Verify $g$ is non-degenerate and \textbf{return False} if not
         \For{each open circuit $F$ in $L$}
         \State $s \gets$ number of costly gates in $F$
         \State{label the open nodes of $F$
           by an arbitrary permutation of $x_1, \dots , x_n$}
         \For{each implicit splice code $E$ of length at most $m + s$}
         \State $d \gets$ number of combiners recorded in $E$
         \For{each $a_1, \ldots a_d \in \mathbb{N}$ s.t. $\sum_{i = 1}^d a_i = m$}
        \Comment{Valid distribution of variables to $Y$-trees}
        \For{each $d$-tuple of read-once formulas
          with $a_1, \dots a_d$ open nodes}
        \State{label the open nodes of each read-once formula
          by an arbitrary permutation of $y$-inputs}
        \State{insert the read-once formulas into combiner instructions of $E$ in lexicographic order}
        \State $\Tilde{G} \leftarrow \texttt{Decode}(F, E)$ 
        \If{$\texttt{tt}(\Tilde{G}) \simeq \texttt{tt}(g)$}
        \Comment{Test using the procedure of Theorem \ref{thm:tt-iso}}
        \State \Return \textbf{True}
        \EndIf
        \EndFor
         \EndFor
         \EndFor
         \EndFor
         \State \Return \textbf{False}
     \end{algorithmic}
 \end{algorithm}

\begin{theorem}\label{thm:ckt-se-solver}
    Given $L$ is a list with a representative from every optimal open circuit class of $f$, Algorithm \ref{alg:ckt-se-solver} decides the $f$-Simple Extension Problem in time $|L| \cdot 2^{O(\ell(s+m))}$ where $\ell$ is the maximum fanout of any node in any circuit in $L$ and $s = CC(f)$.
\end{theorem}

\begin{proof}
    We first prove completeness, i.e. that if $g$ is a simple extension then Algorithm \ref{alg:ckt-se-solver} returns true. By definition of simple extension, the checks in steps $1$ and $2$ pass. Since $g$ is a simple extension, $CC(g) = s + m$; fix any optimal circuit $G$ for $g$. By Theorem \ref{thm:splice-codes}, there is an explicit splice code $\hat{E}$ and circuit isomorphism class representative $\hat{F}$ such that $\texttt{Decode}(\hat{F}, \hat{E})$ produces a circuit $\hat{G}$ isomorphic to $G$. Let $\hat{T}_1, \ldots \hat{T}_t$ be the explicit $Y$-trees in $\hat{E}$. During some iteration of the algorithm (1) the implicit splice code will be consistent with $E$ and (2) $F$ and $T_1, \ldots T_t$, the chosen open read-once formulas, will be isomorphic to $\mathrm{open}(\hat{F}), \mathrm{open}(\hat{T}_1), \ldots \mathrm{open}(\hat{T}_t)$ (where $\mathrm{open}(C)$ is circuit $C$ with its input labels stripped away). In this iteration, let $F'$ and $E'$ be the arbitrary explicit completion of the open circuit $F$ and implicit splice code $E$ chosen by the algorithm. Let $G' = \texttt{Decode}(F', E')$ and observe $\mathrm{open}(G')$ is isomorphic to $\mathrm{open}(\hat{G})$. Thus there is a permutation of the variables of $G'$ such that the resulting circuit is isomorphic to $\hat{G}$. By part (2) of Observation \ref{obs:permuting-variables-truth-table-isomorphism}, the function computed by $G'$ is truth-table isomorphic to $g$ and thus the algorithm accepts.
    
    For soundness, observe that if the algorithm accepts then it passed step $1$ and $2$ which demonstrate the existence of a key and non-degeneracy of $g$. Finally in steps 12 and 13, the algorithm must have constructed a constant-free circuit of size $s+m$ which computes a function truth-table isomorphic to $g$. Applying part (1) of Observation $\ref{obs:permuting-variables-truth-table-isomorphism}$ we can transform this circuit into a circuit for $g$ by relabeling inputs. This constant-free circuit of size $s+m$ for $g$, in conjunction with the fact $f$ is non-degenerate and the existence of a key, guarantees $CC(g) = s + m$ by Lemma \ref{lem:no-false-witnesses}. Therefore, by definition, $g$ is a simple extension of $f$.

    For the running time, we first observe that steps $1$ and $2$ take $2^{O(s+m)}$ since the algorithm just verifies whether one of the $2^m$ restrictions yields $tt(f)$ (and $s = \Omega(n)$) and verifies for each of the $m$ extension variables there is an assignment where flipping the value of the variable changes the output of $g$. There are then $|L|$ iterations where each run in time $2^{O(\ell(s+m))}$ where $\ell$ is the maximum fanout of any node in any circuit in $L$ as the algorithm builds this number of explicit splice codes. The algorithm tries every explicit splice code where $y$-inputs are labeled arbitrarily (say, in increasing order). The number of such codes is $\sum_{d=1}^{m}r(d)$ where $r(d)$ is the number of such splice codes with $d$ combiners. We see 
    \[r(d) = \sum_{\substack{a_1 + \ldots + a_d  = m \\ a_i \in \mathbf{N}^+ }}\,\,\prod_{j=1}^{d} t(a_j),\] where $t(a_j)$ is the number of read-once formulas with $a_j$ open nodes.

    We first bound $t(a_j)$. A number of read-once formula with $a_j$ open nodes corresponds to the number of rooted binary trees where each internal node is labeled $\land$ or $\lor$ and each edge is weighted $0$ or $1$ (representing if there is a negation between those two nodes). The number of unlabeled rooted binary trees with $n$ leaves is $C_{n-1}$---the $(n-1)^{\text{th}}$ Catalan number \cite{van_Lint_Wilson_1992}. Thus $t(a_j) = C_{a_j-1} \cdot 2^{a_j - 1} \cdot 2^{2(a_j - 1) + 1} = C_{a_j-1} \cdot 2^{2{a_j}}$ where the latter factors are from labeling the internal nodes and edge weights (including the output edge) respectively. Since $C_{a_j-1} < 4^{a_j}$, we have $t(a_j) < 2^{5{a_j}}$ \cite{van_Lint_Wilson_1992}. Since the $a_j$ sum to $m$, then $\prod_{j=1}^d t(a_j) < 2^{5m} = 2^{O(m)}$. The number of solutions to $a_1 + \ldots + a_d = m$ where $a_j > 0$ is ${m + d - 1 \choose d - 1}$ \cite{Feller1968}. Notice that ${m + d - 1 \choose d - 1} \leq \sum_{i=0}^{m+d-1}{m + d - 1 \choose i} = 2^{m+d-1}$ \cite{Feller1968}. Since $k < m$ then this is $2^{O(m)}.$ Therefore $r(d) < 2^{O(m)}$ and thus the number of explicit splice codes is at most $m 2^{O(m)} = 2^{O(m)}$.

    Once an explicit splice code is constructed, constructing $\tilde{G}$ with the decoder takes $\mathrm{poly}(s+m)$ time and testing truth table isomorphism takes $O((s+m)^2) \cdot 2^{O(n+m)} \cdot 2^{O(n+m)} = 2^{O(n+m)}$ as the algorithm has to write the truth-table and then run the algorithm from Theorem \ref{thm:tt-iso}. Since $s = \Omega(n)$ this takes $2^{O(s+m)}$. Overall the running time is therefore $|L| \cdot 2^O(\ell \cdot (s + m)).$
\end{proof}

In general, $|L|$ may be exponential in $|tt(g)|$, $\ell$ may be $\omega(1)$, and $s$ may be superlinear in $n$. As such Algorithm \ref{alg:ckt-se-solver} is not a polynomial time algorithm for the Simple Extension problem in general. However, as noted in Corollary \ref{cor:XOR-params}, these parameters for $\XOR_n$ are tractable. This yields our main theorem:
\begin{corollary}\label{cor:xor-se-is-easy}
    The Simple Extension Problem with the base function $f = \XOR_n$ is in $\P$
\end{corollary}

\section{Optimal \texorpdfstring{$\XOR$}{XOR} Circuits Are Binary Trees of \texorpdfstring{$\XOR_2$}{XOR2} Sub-Circuits}
\label{sec:opt-XOR-ckt}

Algorithm \ref{alg:ckt-se-solver} of Section \ref{sec:xor-se-solver} can only rule out a particular $f$ from showing hardness for $\MCSP$ via $\SEP{f}$ when the optimal circuits for $f$ are well-understood. It requires an exact characterization of the set of all optimal circuits---a requirement that currently is rarely fulfilled.

One of the earliest studied explicit Boolean functions, $\XOR$, is one of the only functions for which we even have partial fulfillment of this ornery prerequisite. Recall the definition $\XOR$, the parity function.

\begin{definition}[$\XOR$]\label{def:xor}
    For $x \in \{0,1\}^n$, 
        \(
            \XOR_n(x) = 
            \begin{cases} 
                1 & \text{ if an odd number bits of } x \text{ are } 1 \\
                0 & \text{otherwise}.
            \end{cases}
        \)
\end{definition}

Schnorr proved one of the first explicit circuit lower-bounds on $\XOR$.

\begin{theorem}[\cite{Schnorr74}]
\label{thm:schnorr}
  $\XOR_n$ requires at least $3(n-1)$ gates in the DeMorgan basis.
\end{theorem}

This lower-bound in fact has a matching upper bound and the construction is straightforward: take any binary tree with $n$ leaves labeled $x_1, \ldots, x_n$ where the $n-1$ interior nodes have been labeled by $\oplus$ and replace the $\oplus$ nodes with any circuit of size $3$ that computes $\XOR_2$. Furthermore, since $\neg$ gates do not count towards the circuit size, any circuit for $\XOR_n$ can easily be transformed into an equal size circuit computing $\neg\XOR_n$ and vice versa. Combining these observations yields:

\begin{corollary}[\cite{Schnorr74,Wegener1987}] 
\label{cor:opt-xor-size}
  A circuit $C$ computing $(\neg) \XOR_n$ is optimal if and only if $|C| = 3(n-1)$.
\end{corollary}

However, this leaves open \emph{how} those $3(n-1)$ gates can be arranged. Are there optimal circuits for $\XOR$ that do not follow the upper-bound construction? Previous work showed that in $\cR$, all optimal circuits follow this construction: optimal $\XOR_n$ circuits are binary trees of $(n-1)$ $\XOR_2$ sub-circuits \cite{Kombarov2011}. In this section, we extend this characterization, as seen in Figure \ref{fig:XOR-true-shape}, to $\cD$.

\begin{theorem*}[Informal Statement of Theorem \ref{thm:XOR-structure}]
\label{thm:XOR-structure-informal}
    Optimal $(\lnot)\XOR_n$ circuits in $\cD$ partition into trees of $(n-1)$ $(\lnot)\XOR_2$ sub-circuits.
\end{theorem*}

\begin{figure}[b]
\includegraphics[width=\textwidth]{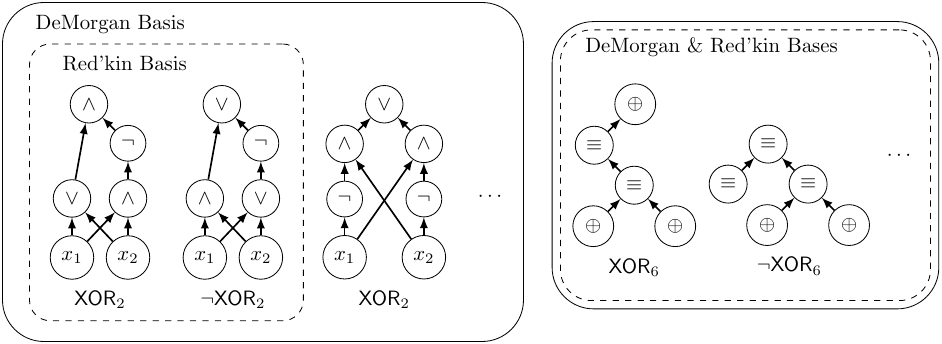}
\label{fig:xor-shape}
\end{figure}

This structural characterization implies the crucial combinatorial parameter bounds that are required to apply Algorithm \ref{alg:ckt-se-solver} and therefore rule out the possibility of proving hardness of $\MCSP$ via $\SEP{\XOR}$.

\begin{corollary}
\label{cor:XOR-params}
    The following properties hold for $\XOR_n$, where $n \geq 1$
    \begin{itemize}
        \item The size of optimal circuits computing $\XOR_n$ is linear (Corollary \ref{cor:opt-xor-size}).

        \item The maximum fan-out of such circuits is a constant. 

        \item The number of optimal circuits for $\XOR_n$, up to permutation of variables, is $2^{O(n)}$.
    \end{itemize}
\end{corollary}

Before we prove our main result of this section, we introduce some auxiliary facts that we will repeatedly require.

\subsection{Basic Properties of \texorpdfstring{$\XOR$}{XOR}}
We first give two basic facts about $(\neg)\XOR_n$ that are immediate consequences of the definition.

\begin{fact}[$(\neg)\XOR$ is Fully DSR] \label{fact:xor-dsr}
  $\XOR_n$ is \emph{fully downward self-reducible}, i.e. for any input $x \in \{0,1\}^n$, any non-empty sets $S$ and $T$ partitioning $[n]$,
  \[
    \XOR_n(x) = \XOR_2(\XOR_{|S|}(x_S), \XOR_{|T|}(x_T))
  \]
where $x_S = \{x_i : i \in S\}$ and $x_T = \{x_i : i \in T\}$. Furthermore, this means for any assignment $\alpha_S$ of variables in $x_S$, $\XOR_n(x)|_{\alpha_S} = (\neg)\XOR_{|T|}(x_T)$. The same is also true of $\neg\XOR_n(x)$
\end{fact}

\begin{fact}[All Subfunctions of $(\neg)\XOR$ are Non-Degenerate] \label{fact:xor-non-degeneracy}
  $(\neg)\XOR_n$ not only depends on all of its inputs but it is also \emph{maximally sensitive}, i.e. for all $i \in [n]$, for inputs $x \in \{0,1\}^n$, $(\neg)\XOR_n(x) \neq (\neg)\XOR_n(x \oplus e_i)$.
\end{fact}

Combining Facts \ref{fact:xor-dsr} and \ref{fact:xor-non-degeneracy}, if we substitute for a single variable in a $(\neg)\XOR_n$ circuit where $n \geq 2$ are guaranteed to get a circuit which computes $(\neg)\XOR_{n-1}$. Applying the tight version of Schnorr (Corollary \ref{cor:opt-xor-size}) we see that subsequent simplification cannot remove more than three costly gates. Formally,

\begin{corollary}[Elimination Rate Limit for Optimal $\XOR$ Circuits]\label{cor:XOR-single-bit-at-most-three}
  Let $C$ be an optimal circuit computing $(\neg)\XOR_n$ where $n \geq 2$. Let $C'$ be the circuit after substituting $x_i = \alpha$ for some $i \in [n]$ and $\alpha = \{0,1\}$ and applying simplifying. We have that $|C'| \geq |C| - 3$ and $C'$ is not constant.
\end{corollary}

We will repeatedly apply this corollary in our proof: deviation from the prescribed structure will often allow us to substitute and remove more than three distinct binary gates. The other main source of contradictions will be substitutions and rewrites that disconnect inputs (violating Fact \ref{fact:xor-non-degeneracy}) or that leave inputs with exactly one costly successor. This violates the fact that $\XOR_n$ reads each of its inputs twice.

\begin{lemma}[$(\lnot)\XOR$ is Read-Twice (Folklore)] 
\label{lem:xor-is-read-twice} 
  Let $C$ be a normalized optimal circuit computing $(\neg)\XOR_n$ where $n \geq 2$. The fanout of every variable $C$ is exactly 2.
\end{lemma}

\begin{remark}
    In the proofs of this section, we will omit negations whenever reasonable, since we are more interested in the binary gates which solely contribute to circuit size in $\cD$. We write $\alpha$ is a \emph{costly successor} to $\beta$ to mean $(\neg)\beta$ is an input to binary gate $\alpha$.
\end{remark}

\begin{proof}
  Let $C$ be an optimal normalized circuit computing $(\neg)\XOR_n$ for arbitrary $n$.  Suppose there is a variable $x_i$ whose fanout is not $2$. There are two cases: (1) the fanout of $x_i$ is $1$ and (2) the fanout of $x_i$ is at least $3$.

  (1) Take $\alpha$, the assumed unique costly successor of $x_i$ and let $\beta$ be the input to $\alpha$. Let $X'$ be the set of input variables that $\beta$ depends on (i.e. that are present in the sub-circuit rooted at $\beta$). Observe that $x_i \not\in X'$. Consider the truth-table of the sub-circuit rooted at $\beta$ and observe it must be a non-constant function of the variables of $X'$. If it were constant, then we could replace $\beta$ with this constant and remove $\alpha$ from the circuit via a gate elimination rule, reducing the size of the circuit and violating optimality. Therefore, there is an assignment of the variables in $X'$ such that if we substitute and simplify, eventually a constant will feed into $\alpha$ that allows us to remove it with a fixing rule. This disconnects $x_i$ from the circuit, which violates Fact $\ref{fact:xor-non-degeneracy}$ since we did not substitute for $x_i$ and thus the the resulting parity function must still depend on it.

  (2) Suppose $x_i$ has more than two costly successors. Let $\alpha_1, \alpha_2$ and $\alpha_3$ be three of these and without loss of generality assume these are indexed in descending depth. Observe that $(\neg)\alpha_3$ is not the output of the circuit as otherwise we could substitute $x_i$ to be some constant that fixes $\alpha_3$ and make the circuit constant. This would violate the rate limit on eliminations for optimal $\XOR$ circuits (Corollary \ref{cor:XOR-single-bit-at-most-three}). Let $\beta_3$ be the costly successor of $\alpha_3$ and notice that, since $\alpha_1, \alpha_2$ and $\alpha_3$ are in descending depth order, $\beta_3$ is a new distinct gate. Thus if we substitute $x_i$ to fix $\beta_3$ and simplify, we can remove $\alpha_1, \alpha_2, \alpha_3$ and $\beta_3$ with a passing or fixing rule applying to $\beta_3$ after applying a fixing rule to $\beta_3$. This violates Corollary \ref{cor:XOR-single-bit-at-most-three}.
  
  Both cases reach a contradiction and therefore every input has two costly successors in any normalized optimal circuit computing $(\neg)\XOR_n$.
\end{proof}

\subsection{Optimal \texorpdfstring{$(\lnot)\XOR_n$}{(¬)XOR} Circuits Are Binary Trees of \texorpdfstring{$(\lnot)\XOR_2$}{(¬)XOR2} Blocks}
We now have the tools required to show that binary trees of optimal $(\neg)\XOR_2$ subcircuits are the \emph{only} optimal circuits for computing $\XOR_n$. Formally,

\begin{theorem}
    \label{thm:XOR-structure}
    Optimal $(\neg)\XOR$ circuits \emph{partition into} trees of $(\neg)\XOR_2$ sub-circuits --- even when NOT gates are free.  Formally, for every circuit $C$ with the minimum number of AND,OR gates computing $\XOR_n$, there is a partition of the gates of $C$ into $(n-1)$ blocks together with a multi-labelling of each wire $w$ in $C$ by tuples $\langle i,t \rangle$ where $t \in \{\mathsf{in}, \mathsf{out}, \mathsf{core} \}$ describes the role that $w$ plays in block $i$, such that:
    \begin{enumerate}
    \item Each block is a three-gate $\XOR_2$ sub-circuit, with distinguished $\mathsf{input}$, $\mathsf{output}$, and $\mathsf{core}$ wires.
    \item The $\mathsf{input}$ wires of every block are also the output wires of a different block or the input gates.
    \item Contracting all the $\mathsf{core}$ wires of $C$ results in a binary tree.
    \end{enumerate}
\end{theorem}

The proof will proceed via induction, however most of the work will be in proving that we can find a $(\lnot) \XOR_2$ block in any $\XOR_{n}$ circuit which we can ``peel off'' with a single variable substitution. The resulting circuit will compute $\XOR_{n-1}$ allowing us apply our inductive hypothesis in order to get a partition we can lift back up to the original circuit. For this reason we separate this out as a lemma.

\begin{lemma}
\label{lem:xor-has-blocks}
    Let $C$ be a circuit computing $\XOR_n$ for $n\geq 3$. There exists two inputs $x_i$ and $x_j$ that feed into a block $B$ in $C$ as described in Theorem \ref{thm:XOR-structure}.
\end{lemma}
\begin{proof}
    Let $C$ be an optimal normalized circuit computing $\XOR_n$ where $n \geq 3$. We first identify two variables which will be inputs to block $B$ (which we will later prove that $B \equiv (\neg)\XOR_2$). Let $\alpha$ be a maximum depth binary gate of $C$. As in the proof of Theorem \ref{thm:schnorr}, we know that $\alpha$ must read two distinct inputs $x_i$ and $x_j$ for some $i,j \in [n]$ since otherwise we can apply a simplification rule and remove $\alpha$, which contradicts that it is an optimal normalized circuit. Our local view of $\alpha$ can be seen in Figure \ref{fig:XOR-1}, where dashed arrows indicate one (or more) adjacent binary gates whose existence has not yet been established.

    \begin{figure}[h]
        \centering
        \includegraphics[]{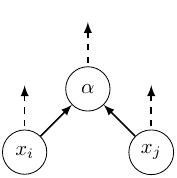}
        \caption{The local neighborhood of $\alpha$}
        \label{fig:XOR-1}
    \end{figure}

    By Lemma \ref{lem:xor-is-read-twice}, we know that both $x_i$ and $x_j$ have exactly two costly successors. Let $\beta_i$ be the other successor of $x_i$. We first observe that neither $\alpha$ nor $\beta_i$ can be the output of the circuit: if they were, some substitution of $x_i$ would fix the output and leave the resulting circuit constant, violating Corollary \ref{cor:XOR-single-bit-at-most-three}. Hence, $\alpha$ and $\beta_i$ must each feed into at least one more costly gate. Let $\nu$ be a costly gate fed by $\alpha$. We remark that $\nu$ could be $\beta_i$ and hence in Figure \ref{fig:XOR-2}, we denote $\nu$ with a dashed node until we establish $\nu \neq \beta_i$.

    \begin{figure}[h]
    \centering 
    \includegraphics[]{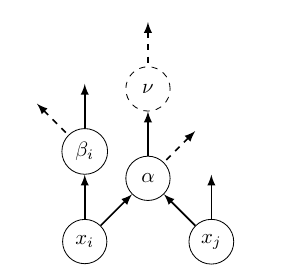}
    \caption{Our view after establishing $\beta_i \neq \alpha$ and neither are the output}
    \label{fig:XOR-2}
    \end{figure}
    
    We first show that $\nu$ is the only successor of $\alpha$. Towards a contradiction, suppose $\alpha$ also feeds other gates besides $\nu$. If it feeds more than one other gate, then substituting $x_i$ to fix $\alpha$ immediately eliminates more than three gates. Hence, it can at most feed one more gate $\nu'$. In particular, we must analyze two cases depending on whether or not $\beta_i$ is distinct from $\nu$ and $\nu'$. These cases can be seen in Figure \ref{fig:XOR-3}. In both cases, we can eliminate more than $3$ costly gates, violating Corollary \ref{cor:XOR-single-bit-at-most-three}.

    \begin{figure}[h]
    \centering
    \captionsetup[subfigure]{justification=centering}
    \captionsetup[subfigure]{labelformat=empty}
    \begin{subfigure}{.33\textwidth}
        \includegraphics[]{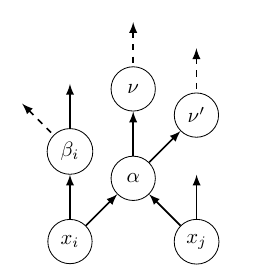}
        \caption{(a) $\beta_i$ is distinct from $\nu'$}
    \end{subfigure}
    \begin{subfigure}{.33\textwidth}
        \includegraphics[]{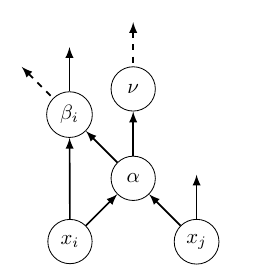}
        \caption{\hspace{-0.2cm}(b) $\beta_i$ is not distinct from $\nu'$}
    \end{subfigure}
    \caption{The two cases if $\alpha$ feeds two gates}
    \label{fig:XOR-3}
    \end{figure}

    \begin{enumerate}
        \item Assume $\beta_i$ is distinct from $\nu$ and $\nu'$. Fixing $\alpha$ with $x_i$ will eliminate four distinct gates, $\alpha, \beta_i, \nu,\nu'$.
        \item Assume $\beta_i$ is not distinct from $\nu$ and $\nu'$. Without loss of generality, assume $\beta_i = \nu'$. Substituting $x_i$ to fix $\alpha$ also fixes $\beta_i$ since both of its inputs are now constants. If $\beta_i$ has any costly successors besides $\nu$ then we are done, since that eliminates at least four gates ($\alpha, \nu, \beta_i,$ and any distinct successor). If $\beta_i$ only feeds into $\nu$, then $\nu$ also becomes fixed. Hence, $\nu$ cannot be the output of the circuit and thus it must have at least one costly successor that is distinct from $\alpha, \nu$ and $\beta_i$. This successor can also be eliminated from our substitution.
    \end{enumerate}
    
    We now show that $\nu$ and $\beta_i$ are distinct. Assume otherwise, then fixing $\alpha$ by substituting $x_i$ also fixes $\beta_i = \nu$, which we know is not the output of the circuit. This eliminates at least 3 (and therefore exactly three) gates ($\alpha, \beta_i,$ and $\beta_i$'s successor which we will call $\upsilon$) and results in an optimal $\XOR_{n-1}$ circuit. We now consider the other gate fed by $x_j$ which we will call $\beta_j$. Again, there are two cases as seen in Figure \ref{fig:XOR-5}: either $\beta_j$ is distinct from $\upsilon$ or they are the same.

    \begin{figure}[h]
    \centering
    \captionsetup[subfigure]{justification=centering}
    \captionsetup[subfigure]{labelformat=empty}
    \begin{subfigure}{.3\textwidth}
        \includegraphics[]{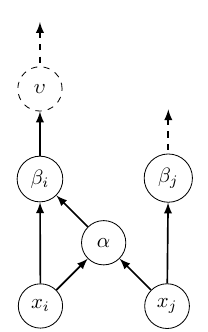}
        \caption{}
    \end{subfigure}
    \begin{subfigure}{.3\textwidth}
        \includegraphics[]{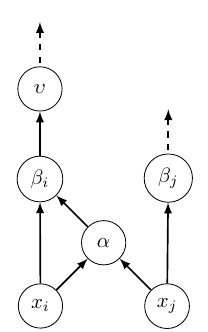}
        \caption{\hspace{-0.5cm}(a) $\beta_j$ is distinct from $\upsilon$}
    \end{subfigure}
    \begin{subfigure}{.3\textwidth}
        \includegraphics[]{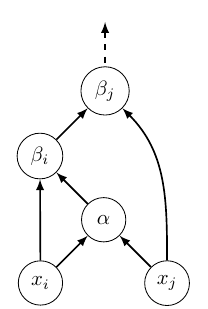}
        \caption{\hspace{-0.5cm}(b) $\beta_j$ is not distinct from $\upsilon$}
    \end{subfigure}
    \caption{The two cases if $\nu = \beta_i$}
    \label{fig:XOR-5}
    \end{figure}

    If $\beta_i$ does not also feed into $\beta_j$, then $x_j$ feeds only $\beta_j$ in the resulting circuit contradicting Lemma $\ref{lem:xor-is-read-twice}$. If $\beta_j$ is fed by $\beta_i$, then consider instead substituting $x_j$ to fix $\beta_j$. Observe that $\beta_j$ is not the output of the circuit, and hence fixing it eliminates four gates: $\alpha$, $\beta_i$, $\beta_j$, and $\beta_j$'s costly successors. In both cases we reach a contradiction, hence $\beta_i \neq \nu$. Furthermore, notice that $\beta_i$ cannot be the output gate and it must have only one costly successor. Otherwise, using $x_i$ to fix $\beta_i$ would eliminate at least four gates: $\alpha, \beta_i$, and the costly successors of $\beta_i$. Hence, we arrive at the local view around $h$ as shown in Figure \ref{fig:XOR-6}.

    \begin{figure}[h]
        \centering
    \includegraphics[]{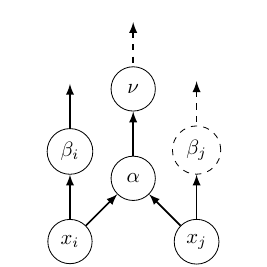}
    \caption{Our view after establishing $\nu \neq \beta_i$}
    \label{fig:XOR-6}
    \end{figure}

    Now, if we fix $\alpha$ using $x_i$, then $\alpha, \beta_i,$ and $\nu$ are eliminated and thus must be the only gate that are eliminated. 
    Thus, if $\beta_j$, the other successor of $x_j$ is distinct from these three gates, then $x_j$ is left with one costly successor in the resulting circuit computing $\XOR_{n-1}$ thus violating Lemma \ref{lem:xor-is-read-twice}. Hence $\beta_j$ must be one of $\beta_i$ or $\nu$. We will show it must be $\beta_i$.

    Suppose, for the sake of contradiction, that $\beta_j = \nu$. We can see that $\beta_i$ must be fed by $\nu$, otherwise we could substitute $x_j$ to fix $\alpha$ and the resulting $\XOR_{n-1}$ circuit only reads $x_i$ once. Now that $\beta_i$ is fed by $\nu$, if we substitute $x_i$ to fix $\beta_i$, then this eliminates $\beta_i, \alpha$ and $\nu$ and disconnects $x_j$ entirely from the resulting circuit, but the resulting $(\neg)\XOR_{n-1}$ circuit must still depend on $x_j$. 
    
    We now have that $\beta_j = \beta_i$---let us call the gate $\beta$. Recall $\beta$ has exactly one costly successor, and if this costly successor isn't $\nu$ and is another gate then fixing $\alpha$ with $x_j$ eliminates three gates $(\alpha, \beta, \nu)$ but leaves $x_j$ with exactly one costly successor it inherited from $\beta$. This is a contradiction, and hence $\beta$ must feed $\nu$ which yields a block $B$ as shown in Figure \ref{fig:XOR-8}.
\end{proof}

    \begin{figure}[h]
        \centering
        \includegraphics[]{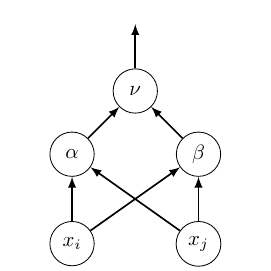}
    \caption{The local area around $x_i$ and $x_j$ form a block B.}
    \label{fig:XOR-8}
    \end{figure}

We now provide the proof of Theorem \ref{thm:XOR-structure} 

\begin{proof}[Proof of Theorem \ref{thm:XOR-structure} using Lemma \ref{lem:xor-has-blocks}]
    We prove this via induction. For $n=1$ and $n=2$ the theorem trivially holds: $(\neg)x_i$ is the unique normal optimal circuit for $(\neg)\XOR_1$ and normal optimal $(\neg)\XOR_2$ circuits trivially define a single block.

    Assume the statement holds for some $k-1 \geq 2$. Let $C$ be a normal optimal circuit computing $(\neg)\XOR_k$ for $k \geq 3$. By Lemma \ref{lem:xor-has-blocks}, we know $C$ must have a block $B$ that is fed by two distinct inputs $(\neg)x_i$ and $(\neg)x_j$. Let the three gates be $\alpha, \beta,$ and $\nu$ where $(\neg)\nu$ is the output gate of $B$. We label the wires from $(\neg)x_i$ and $(\neg)x_j$ as $\texttt{in}$, the wires from $(\neg)\nu$ as $\texttt{out}$ and all other internal wires of $B$ as $\texttt{core}$. We first need to partition the rest of the circuit into blocks and ensure that the two $\texttt{out}$ wires are $\texttt{in}$ wires to the same block in the rest of $C$.

    If we substitute $x_i = b$ to fix $\alpha$ and simplify, we see that $x_j$'s successors are replaced by $(\neg)\nu$'s successors and that the three costly gates in $B$ have been eliminated. Therefore the circuit computes $(\neg)\XOR_{k-1}$ and is optimal. Applying the inductive hypothesis, we can partition the remaining circuit into blocks which we lift back to the original. Notice that $x_j$'s new successors are in the same block and therefore in $C$, $(\neg)\nu$'s successors are also in the same block as desired.

    It remains to prove that $B$ computes $(\neg)\XOR_2$. If we substitute $x_i = b$ to fix $h$ and rewrite as above we see that $B$ reduces to $(\neg)x_j$, i.e. $B(b, x_j) = (\neg)x_j$. We argue that if we instead substitute $x_i = 1-b$, then $B$ would also reduce to $(\neg)x_j$. It cannot reduce to a constant as otherwise $C$, which now computes $\XOR_{k-1}$ does not depend on $x_j$, contradicting Fact \ref{fact:xor-non-degeneracy}. We also observe $B$ cannot reduce to just $x_j$ in both cases (or $\neg x_j$ in both cases), as otherwise we could replace $B$ with $x_j$ (or $\neg x_j$) and $C$ would still correctly compute $\XOR_n$ with three fewer gates --- contradicting that $C$ is optimal. Therefore $B(1-b, x_j) = \neg B(b, x_j)$ and the only binary Boolean functions that satisfy these two equations are $\XOR_2$ and $\neg\XOR_2$.
\end{proof}

\end{document}